\title{Finite Volume Spaces  ~and~ Sparsification}
\author{Ilan Newman 
$~ ~ ~ ~$ Yuri Rabinovich \\ \\
Department of Computer Science, University of
Haifa, Haifa, Israel. \\
E-mails: {\tt \{ilan,yuri\}@cs.haifa.ac.il}}

\date{\today}
\documentclass[11pt]{article}
\usepackage{amssymb,fullpage}
\usepackage{xspace}
\usepackage{latexsym}
\usepackage{times}
\usepackage{amsfonts}
\usepackage{amsmath}

\newcommand{\Cl}[1]{\rm Col(#1)} % boundary as a set of vectors.
\newcommand{\ignore}[1]{}
\def \qed {\hspace*{0pt} \hfill {\quad \vrule height 1ex width 1ex depth 0pt}
 \medskip}
\def \e {\epsilon}
\def \d {\delta}
\def \P {\mathrm Pr}
\def \C {\mathcal C}
\def \F {\mathcal F}

\def \trk {\rm trk}
\def \srk {\rm rank^{*}}
\def \rank{\rm rank}
\def \span{\rm span}
\def \v {\rm vol}
\def \K {K^{(2)}_n}
\def \Link {\rm link}
\def \Col {\rm Col}
\def \Cp {\rm Cap}
\def \Row {\rm Row}
\def \conv{\rm conv}
\def \av{\rm av}
\def \dist {\rm dist}
\def \supp {\rm supp}
\def \cone {\rm cone}

\newenvironment{proof}{\par\noindent{\bf Proof}\quad}{  $\qed$}

\newcommand{\R}{\ensuremath{\mathbb R}}
\newcommand{\Z}{\ensuremath{\mathbb Z}}
\newcommand{\N}{\ensuremath{\mathbb N}}
\newcommand{\E}{\ensuremath{\mathrm E}}
\renewcommand{\S}{\ensuremath{\mathrm S}}

\newtheorem{theorem}{Theorem}%[section]

\newtheorem{definition}{Definition}
\newtheorem{claim}{Claim}[section]

\newtheorem{lemma}{Lemma}[section]

\newtheorem{corollary}{Corollary}[section]

\begin{document}
%\setcounter{page}{10}
%\maketitle

\begin{titlepage}
\maketitle 
\def\thepage {} % Kill pagenumbering
\thispagestyle{empty}

\begin{abstract}
We introduce and study finite $d$-volumes - the high dimensional generalization of
finite metric spaces. Having developed a suitable combinatorial
machinery, we define $\ell_1$-volumes and show that they contain
Euclidean volumes and hypertree volumes. We show that they can
approximate any $d$-volume with $O(n^d)$ multiplicative distortion. On
the other hand, contrary to Bourgain's theorem for $d=1$, there exists a
$2$-volume that on $n$ vertices that cannot 
be approximated by any $\ell_1$-volume with distortion smaller than
$\tilde{\Omega}(n^{1/5})$. 

We further address the problem of $\ell_1$-dimension
reduction in the context of $\ell_1$ volumes, and show that this
phenomenon does occur, although not to the same striking degree 
as it does for Euclidean metrics and volumes.
In particular, we show that any $\ell_1$ metric on $n$ points can be
$(1+ \epsilon)$-approximated by a sum of $O(n/\epsilon^2)$ cut
metrics, improving over the best previously known bound of 
$O(n \log n)$ due to Schechtman.

In order to deal with dimension reduction, we extend the techniques and ideas 
introduced by Karger and Bencz{\'u}r, and Spielman et al.~in the context of graph 
Sparsification, and develop general methods with a wide range of
applications.
%%%%%%%%%%%%%%%%%%%%%%%%%%
\ignore{%old abstract
We show that for any $\epsilon >0$ and any $\ell_1$ metric on $n$ points $d$
there is a metric $d'$ that approximates $d$ by at most $1+ \epsilon$ and
$d'$ can be written as a non-negative combination of $O(n/\e^2)$ cuts
metrics. This implies that any $\ell_1$ metric on $n$ points
can be embedded into $O(n/\e^2)$-dimensional $\ell_1$-space, with
distortion bounded by $1+ \e$.

We then show that the methods used have a much wider range of applicability,  
and develop basic tools to analyze their performance in the general settings.

We suggest a novel abstract definition of volumes, $\ell_1$-volumes, 
and cut-volumes, the counterparts of well known concepts in theory of
finite metric spaces. We proceed to establish their basic properties,
and in particular, establish a cut-dimension reduction result. 
}
%%%%%%%%%%%%%%%%%%%%%%%%%%%%
\end{abstract}

{\bf ACM classes:} G.2.0.; G.2.1; F.2.2

\end{titlepage}

%}%ignore title page
\newpage
\pagenumbering {arabic} % back to usual

%%%%%%%%%%%%%%%%%%%%%%%%%%%%%%%%  definition of shortcuts symbols %%%%%%%%%%%%%%%%%%%
\def \noproof
 {\hspace*{0pt} \hfill {\quad \vrule height 1ex width 1ex depth 0pt}}
% Use `noproof' instead of `qed' inside unproven lemmas for better spacing.
\def \xor {\oplus}

\section{Introduction}
This paper has two intertwined storylines. The first is a systematic attempt
to develop a basic theory of {\em finite volume spaces} -  a natural 
generalization of finite metric spaces. The second is an effort to extend the techniques 
and the ideas introduced in~\cite{karger-benczur},~\cite{spielman-sparsifier}, and to make them 
applicable to a wide class of sparsification problems. The synthesis of the two is reached
when the resulting new sparsification methods are successfully applied in the context of 
finite volume spaces, for the $\ell_1$-dimension reduction problem.

The blossoming of the theory of metric spaces in the last two decades affected
both practical and theoretical algorithms design, and also the local theory of normed 
spaces. It developed its own key notions, posed intriguing new problems, and solved many
of these problems using novel methods. There is a rich interplay between 
the theory of finite metric spaces and graph theory. Often the former provides a 
unique prospective on many basic and important graph theoretic notions such as cuts, 
flows, expansion, minors and spanners. Motivated by all this, we introduce the abstract
finite volume spaces, and attempt to use the notions, ideas and methods of finite metrics 
spaces in this more general setting. In doing this, we hope to contribute not only
to the theory of finite volume/metric spaces, but also to the combinatorial theory of 
simplicial complexes. We also get some new geometrical and algorithmical applications.     

The combinatorial theory of simplicial complexes draws much research activity in the recent years, 
as testified, to name but a few, by the studies of random 2-dimensional 
complexes,~\cite{linial-meshulam}, \cite{meshulam-wallach},\cite{babson},\cite{pippenger-surfaces}, 
and the studies of embeddability of $d$-complexes in $\R^n$,~\cite{MTW}. While developing
the theory of finite volume spaces, we naturally arrive at complex-theoretic notions such as
hypercuts, face expansion, and sparse spanners. We establish some of their structural properties,
and present some new constructions.  

The transfer to higher dimension is 
not without difficulties even on the level of basic definitions. E.g., the hypertrees 
(generalizing trees) have numerous distinct definitions, e.g.~\cite{parekh-forestation, adin, 
duke-erdos, lovasz}. Hypercuts (generalizing cuts) remain without explicit definition. 
(A number of possible definition are discussed in this paper. See also the supports of 
coboundaries of~\cite{linial-meshulam}, and the two-graphs of Seidel~\cite{seidel}.)
In a sense, the theory of finite volume spaces helps to make a coherent choice among possible 
conflicting definitions. To clarify the presentation, we make an effort to consistently use 
the language of combinatorics and linear algebra instead of referring to algebraic 
topology. We also try to keep the presentation self-consistent, including in Section~\ref{sec:2}
some basic facts equipped with short proofs.

Having provided the necessary combinatorial background, we embark on systematic
study of finite volumes. In particular, using  hypercuts, we define $\ell_1$-volumes, 
and show that they can be used to approximate any finite volume, and that they contain the 
Euclidean volumes and the hypertree volumes. We show that contrary to Bourgain's 
theorem for $d=1$, there exists a $2$-dimensional volume on $n$ vertices that cannot 
be approximated by any $\ell_1$-volume with distortion smaller than
$\tilde{\Omega}(n^{1/5})$. The best corresponding upper bound we can currently show is
$O(n^2)$. 

The most technically elaborated part of our study of finite $d$-volumes is the
the problem of $\ell_1$-dimension reduction. 

The following is known. For the Euclidean $d$-volumes on $n$ points, the result
of~\cite{Magen} (that extends the famous Johnson-Lindenstrauss Lemma) shows that about 
$O(\epsilon^{-2} \log n)$ dimension will suffice for a $(1+\epsilon)$-faithful representation.
For $\ell_1$-metrics, the elegant lower bound of Brinkman and Charikar~\cite{brikman-charikar} 
(see also Lee and Naor~\cite{lee-naor}) shows that in general, in order to get multiplicative 
distortion $O(1+\epsilon)$ for a small $\epsilon$, one might need many as $n^{0.5}$
dimensions. The best corresponding upper bound is due to Schechtman~\cite{schechtman87}, showing 
that $c_\epsilon n\log n$ dimensions suffice to get a $(1+\epsilon)$ distortion.

We show that $\ell_1$ $d$-volumes can be $(1\pm \epsilon)$-faithfully represented
using $O(n^d\log n/\epsilon^2)$ hypercut $d$-volumes, the high-dimensional analog 
of cut-metrics. This improves the trivial $O(n^{d+1})$ upper bound.
Moreover, for a natural subclass of $\ell_1$ $d$-volumes,
we show a stronger bound of $O(n^d/\epsilon^2)$ of special hypercut $d$-volumes.
Since for $d=1$ all $\ell_1$ metrics belong to this special subclass, we obtain 
an $O(n/\epsilon^2)$ upper bound on the approximate {\em cut dimension} of
any $\ell_1$ metric on $n$ points. This improves on~\cite{schechtman87} in two 
ways: the number of dimensions is smaller, and each dimension is a cut-metric,
a very special case of a line metric.

To deal with the dimension reduction problem, we develop general sparsification methods 
extending the ideas and techniques of~\cite{karger-benczur}
and~\cite{spielman-sparsifier},  
originally aimed for graph sparsification. We believe that the resulting methods are of independent 
theoretical and algorithmical interest. Section~\ref{sec:last} contains a short discussion 
of these methods, as well as an other application to a certain natural problem about  
geometric discrepancy.
%%%%%%%%%%%%%
\ignore{
The paper is organized as follows:

 of the paper: The main parts of the paper are
Sections \ref{sec:vol} in which we define and discuss $d$-volumes and
Section \ref{sec:reduc} devoted to dimension reduction. Section
\ref{sec:2} starts with essentially known standard notions from the
combinatorics of simplicial complexes and proceeds with the additional
combinatorial results needed for the theory of $d$-volumes,
and that are, to our best knowledge, new. We defer some of the proofs
to an Appendix section.  
}
%%%%%%%%%%%%%%%%%%%%%%%%%%%%%%%%%%%%%
%%%%%%%%   FROM HERE
%%%%%%%%%%%%%%%%%%%%%%%%%%%%%%%%%%%%%%

\section{Basics of Combinatorics of Simplicial Complexes}\label{sec:2}
\subsection{Cycles, Hypertrees  and Coboundaries}\label{sec:standard}
Let $V$ be an underlying set of size $n$ and let $K_n^{(d)} = \{\sigma
\subseteq V|~ |\sigma|=d+1 \}$ be the set of all $d$-dimensional simplices
on $V$.
The {\em boundary operator} $\partial$ 
maps a $d$-simplex $\sigma$ to a formal sum over $\Z_2$ of 
the $(d-1)$-subsimplices of $\sigma$ of co-dimension $1$. 
For  a set 
$A \subseteq K_n^{(d)}$, $\partial A$ is defined as 
$\partial A \;=\; \sum_{\sigma   \in A} \partial \sigma$. 
By virtue of $\Z_2$, this formal sum can be
identified with a subset of $K_n^{d-1}$.
It is convenient to think about $\partial$ in terms of the 
 ${n \choose {d}} \times {n \choose {d+1}}$ incidence matrix $M_d$ over $\Z_2$ whose 
rows are indexed by $(d-1)$-simplices, the columns are indexed by $d$-simplices, and
$M_d(\tau,\sigma) = 1$ if $\tau \subset \sigma$, and $0$ otherwise.
%%%%%%%%%%%%%%%%%
\ignore{
%in: inserted
Then for a $d$-simplex $\sigma$, $\partial \sigma$ is identified 
with the column of $M_d$ that is indexed by $\sigma$, namely, that has
a $1$ in every row indexed by $\tau$ for which $\tau \subseteq \sigma$. 
For a set $A$ of $d$-simplices it will also
be convenient to identify $A$, with $\Cl{A}$, the set of columns of
$M_d$ indexed by $\{\sigma \in A\}$. Namely, with this notation
$\partial A$ as vector in ${\mathbb F}_2^{n \choose d}$ can be written as
$\partial A = \sum_{\sigma \in \Cl{A}} \partial \sigma = M_d 1_A$, where $1_A \in \{0,1\}^{n \choose (d+1)}$ is the
characteristic vector of $A$ as a subset of $K_n^{(d)}$.
}
%%%%%%%%%%%%%%%%%%%%
Then, for a set $A$ of $d$-simplices it holds that $M_d 1_A = 1_{\partial A}$.

A {\em $d$-cycle} $Z\subseteq K_n^{(d)}$ is a subset of $d$-simplices that 
vanishes under the boundary operator, i.e., $\partial Z=0$, or $M_d 1_Z = 0$.

Let a (spanning)  {\em $d$-hypertree} be  a maximal acyclic subset of $d$-simplices in
$K_n^{(d)}$. It is easy to verify that like the usual spanning trees, $d$-hypertrees
form a matroid, and therefore are all of the same size. Since the set
of all $d$-simplices containing a fixed vertex $v$ of $V$
is a $d$-hypertree, the size of any $d$-hypertree must be is ${{n-1}\choose {d}}$.
We call $K \subseteq K_n^{(d)}$ {\em homologically connected}, or
(without a risk of confusion with other definitions of connectivity) just {\em connected}
if $K$ contains a $d$-hypertree. 
(The connectivity of $K$ is equivalent to the vanishing of the homology and the cohomology 
groups $H_{d-1}(K)$, $H^{d-1}(K)=0$ over $\Z_2$, where $K$ is treated as a 
simplicial complex containing  all low dimensional simplices on $V$.)

Let $G=G_{d-1} \subseteq K_n^{(d-1)}$ be a subset of $(d-1)$-dimensional 
simplices on $V$.  A {\em $d$-coboundary} $B$ induced by $G$ is  
the sets of all $d$-simplices $\sigma\in K_n^{(d)}$, such that 
the number of $(d-1)$-dimensional faces of $\sigma$ that belong to $G$ 
is odd. I.e., $1_G^T M_d = 1_{B}^T$.  From this definition it is clear the
$d$-coboundaries, like $d$-cycles, form a linear space over $\Z_2$.
A basic relation between the cycles and the coboundaries is:

\begin{claim}\label{cl:cut-cyc}
For any $d$-cycle $Z$ and a $d$-coboundary $B$,  ~$|Z \cap B|$ is even. 
\end{claim}
\begin{proof}
One needs to show that $1_B^T \cdot 1_Z =0$ over $\Z_2$. Let $G$
be the $(d-1)$-complex that induces $B$. Then,
\[
1_B \cdot 1_Z  ~=~ 1_G^T M_d 1_Z ~=~ 1_G^T \cdot \overline{0} ~=~ 0\,,
\]
where $\overline{0}$ is the all-zero vector.
\end{proof}

In fact, the about claim can be taken as an alternative definition of the coboundaries; 
moreover, it suffices to consider only cycles $Z$ of the type $\partial \Delta_{d+1}$, 
i.e., the boundaries of $(d+1)$-simplices on $V$.  

The hypertrees and the coboundaries are related in a complementary manner:
\begin{claim}\label{cl:cut-tree}
$K \subseteq K_n^{(d)}$ is connected iff 
$K\cap B \neq \emptyset$ for any nonempty $d$-coboundary $B$.
\end{claim}
\begin{proof}
We first show that for any hypertree $T$ and any coboundary $B$,
$T \cap B$ is not empty. Indeed, let $G$ be the subset of $K_n^{(d-1)}$
that induces $B$. If $T \cap B$ is empty, $1_G$ is orthogonal to all
the columns of $M_d$ corresponding to $\sigma \in T$. But these columns
span the entire column space of $M_d$, and thus $B$ must be trivial, contrary to
our assumption. Thus, if $K$ is connected, it intersects all the coboundaries.

Assume now that $K$ is not connected, i.e., the columns of $M_d$ corresponding
to $d$-simplices in $K$ do not span the column space. Then, there must
exist a vector $1_G$ orthogonal to all these columns, but not to the entire column
space. The induced $B$ is thus nontrivial, and disjoint with $K$.  
\end{proof}
\ignore{
%in: inserted
The following is basic in the theory of algebraic topology.
\begin{claim}
$M_{d-1}M_d = 0$
\end{claim}
\begin{proof}
Let $A = M_{d-1}M_d$ then for $\sigma \in K^{(d-2)}_n, \tau \in
K^{(d)}_n,~ A_{\sigma,\tau} = |\{\nu \in K_n^{(d-1)}|~ \sigma \subset
\nu \subset \tau \}| ~\bmod(2)$. However, for any such $\sigma, \tau$,
there is either no such $\nu$ or exactly $2$.
\end{proof}

Hence, it follows that $(d-1)$ coboundaries are in the left kernel of
$M_d$ (and in fact, they span this left kernel).
}%ignore

While any $G_{d-1}$ uniquely defines a $d$-coboundary $B$,
the opposite does not hold, and different $G$'s may induce the same $B$. 
In fact, $G$ and $G'$ induce the same $B_d$ iff $G' = G \oplus B_{d-1}$
where $B_{d-1}$ is a $(d-1)$-coboundary.\footnote{
This follows since $M_{d-1}M_d = 0$, and hence any $(d-1)$-coboundary is in
the left kernel of $M_d$. Moreover, comparing the dimensions of the
left kernel of $M_d$ and the space of $(d-1)$-coboundaries, one concludes that the
two are equal. Using the language of the algebraic 
topology, this can be restated as $H^{(d-1)}(K_n^{(d)})=0$, which in turn
follows from the connectedness of $K_n^{(d)}$.}
The ambiguity in choosing $G_{d-1}$ for a given $B$ can be 
removed in the following manner. 
For $X \subseteq K_n^{(d)}$ and $v$ a vertex of $X$, define the 
{\em link} of $X$ with respect to $v$ to be the following $(d-1)$-dimensional 
subcomplex of  $X$:
\[
\Link_v(X) ~=~ \{ \tau \in K_n^{(d-1)}  \;|\;  v\not\in \tau
~~\mbox{and}~~ \{\tau \cup v\} \in X\}.
\]  
\begin{claim}\label{cl:link}
A $d$-coboundary $B$ is induced by $\Link_v(B)$.
Consequently, there is a 1-1 correspondence between the $(d-1)$-dimensional $G_{d-1}$'s 
on $V - \{v\}$, and the $d$-coboundaries $B \subseteq  K_n^{(d)}$.
\end{claim}
\begin{proof}
Let $B'$ be the $d$-coboundary induced by $\Link_v(B)$. Consider first
a $d$-simplex $\sigma$ that contains $v$. Since $\Link_v(B)$ lacks
all the $(d-1)$-faces of $\sigma$ containing $v$, and contains the remaining $(d-1)$-face 
$\tau=\sigma -\{v\}$ iff $\sigma$ is in to $B$, the definition of coboundary $B'$ implies that 
$\sigma \in B'$ iff  $\sigma \in B$. 
Consider next a $d$-simplex $\sigma=(v_1,v_2,\ldots,v_{d+1})$ that does not contain 
$v$. Consider the $d$-boundary of the $(d+1)$-simplex $(v_1,v_2,\ldots,v_{d+1},v)$.
It is a cycle, and all its $d$-faces with exception of $\sigma$ contain $v$.
Since $B'$ and $B$ agree on all these faces, the parity argument 
from Claim~\ref{cl:cut-cyc} implies that they agree on $\sigma$ as well. 
Thus, $B'=B$.
\end{proof}
%%%%%%%%%%%%%%%%%%%%%%%%%%%%
% 
%%%%%%%%%%%%%%%%%%%%%%%%%%%
\subsection{Hypercuts}
%%%%%%%%%%%%%%%%%%%%%%%%%%%
The generalization of cuts in graphs to higher dimensions is not 
straightforward.  Topologists, in view of Claim~\ref{cl:cut-tree},
usually consider the coboundaries to be the proper generalization 
of cuts in graphs. We refine this topological definition, arriving at
a notion that makes a lot of sense also from the volume-theoretic 
perspective (see the Section~\ref{sec:vol} below), as well as from the 
viewpoint of Matroid Theory.

For $A \subseteq  K_n^{(d)}$, define an equivalence relation on $d$-simplices, 
$\sigma_1 \sim \sigma_2 \mod A$, if they are {\em homologous} relatively to $A$.
I.e., there exists a simple $d$-cycle containing $\sigma_1$, $\sigma_2$,
while the rest of its $d$-simplices belong to $A$.
In terms of the matrix $M_d$, it means the following. Let $\Col(X)$ denote the 
set of columns of $M_d$ indexed by $\sigma\in X \subseteq K_n^{(d)}$. Then, 
$\sigma_1 \sim \sigma_2 \mod A$ if
$\;1_{\{\sigma_1\}} - 1_{\{\sigma_2\}} \,\in\, \span\{\Col(A)\}$.
Call a $d$-simplex {\em null homologous} relative to $A$ if there exists a simple 
$d$-cycle containing $\sigma$, while the rest of its $d$-simplices belong to $A$.
Equivalently, $\;1_{ \{ \sigma \} } \,\in\, \span\{ \Col(A) \}$.
%%%
\begin{definition}\label{def:comb-cut}
Call $C \neq \emptyset$, a subset of $d$-simplices, a {\rm (combinatorial) $d$-hypercut} if 
~$(*)$ no $\sigma \in C$ is null homologous relatively to $\overline{C}$; and
~$(**)$ for any $\sigma_1,\sigma_2 \in C$ it holds  that 
$~\sigma_1 \sim \sigma_2 \mod \overline{C}$. 
\end{definition}
%%%%%%
In other words, $C$ is a hypercut iff $\overline{C}$ is maximal unconnected.
This happens to be precisely the definition of the {\em co-circuit} of 
$K_n^{(d)}$  treated as a simplicial matroid. 

In terms of the matrix $M_d$, the Definition~\ref{def:comb-cut}
means the following. 
Let $\Col$ denote the set of columns of $M_d$. Then, $C$ is a hypercut iff 
$\span\{\Col(\overline{C})\} \cap \Col =  \Col(\overline{C})$, and 
the co-dimension of $\span\{\Col(\overline{C})\}$ in $\span({\Col})$
is 1.
%%%%%%%%
\begin{theorem}\label{cl:cut-boundary}
$d$-Hypercuts are precisely the $d$-coboundaries that are minimal
with respect to containment. Moreover, any $d$-coboundary $B$
is a disjoint union of  $d$-hypercuts.
\end{theorem}
\begin{proof} 
The matrix definition of $C$ implies that there exists a vector $y$ 
such that $y \cdot v =0$ for any $v\in \Col(\overline{C})$, and $y \cdot v =1$ 
the rest of the columns.  Thus, a $d$-hypercut is also a $d$-coboundary.

Observe that for a $d$-coboundary $B$ it always holds that 
$\;\span\{\Col(\overline{B})\} \cap \Col =  \Col(\overline{B})$.
If there exists nontrivial $d$-coboundary $B' \subset B$, then the following
strict containments hold,
\[
\span\{\Col(\overline{B})\} \,\subset\, \span\{\Col(\overline{B'})\} 
\,\subset\, \span\{\Col\}\,,
\] 
implying that $\;\span \{\Col(\overline{B})\}$ has co-dimension $>1$, and thus is not 
a hypercut. For the other direction, if $B$ is minimal with respect to
containment,  then for any $\sigma \in B$ it must hold
$\;\span\{\Col(\overline{B} \cup \sigma)\} \,=\, \span\{\Col\}$, 
and thus $\;\span\{\Col(\overline{B})\}$ has co-dimension 1, and therefore 
is a hypercut.

Finally, let $B$ and $B' \subset B$ be coboundaries. Since coboundaries
are closed under addition, $B \setminus B' = B\oplus B'$ is also a coboundary, and thus
$B$ is a disjoint union of two coboundaries. Continuing decomposing these coboundaries,
one arrives at a disjoint union of minimal coboundaries,  i.e., hypercuts.
\end{proof}

The following theorem is analogous to the fact that cutting an edge of a 
spanning tree one obtains a cut.
\begin{theorem}\label{cor:T_sigma}
Let $T$ be a $d$-hypertree, and $\sigma \in T$. Then there exists
a {\em unique} $d$-hypercut $C_{T,\sigma}$ such that \\
$T\cap C_{T,\sigma} = \sigma$. More explicitly, $C_{T,\sigma}$ 
is the set of all the $d$-simplices $\tau$ such that
the {\rm unique} cycle $Z$ created by adding $\tau$ to $T$, contains $\sigma$.
\end{theorem}
\begin{proof} 
Consider the set $S$ of all $d$-simplices whose columns are spanned
by $\Col(T - \{\sigma\})$. Observe that any hypercut disjoint with 
$T - \{\sigma\}$ must also be disjoint with $S$. Let $C=\overline{S}$. 
Observe that $C$ is not empty, as $\sigma \in C$. We claim that $C$ is a hypercut. 
Indeed,~$(*)$ holds by definition of $C$, while~$(**)$ holds since any $d$-simplex 
$\tau$ is null homologous with respect to $T$, and thus, if it is not in $S$, it 
must be homologous to $\sigma$ relatively to $T - \{\sigma\}$. 
\end{proof}

As a corollary of Theorem~\ref{cor:T_sigma} we obtain another definition
of the hypercuts. 
\begin{corollary}\label{thm:blocker}
Let $\mathcal{C}$ be the set of $d$-hypercuts and let $\mathcal{T}$
be the set of $d$-hypertrees. Then, $\mathcal{C}$ is the {\em blocker}
of $\mathcal{T}$, $\mathcal{C} = \mathcal{T}^B$.  That is, every
hypercut intersect every hypertree, and any set $S \subseteq K_n^{(d)}$
with this property that is minimal (with respect to containment) is a hypercut.
\end{corollary}
\begin{proof} 
The statement directly follows from Claim~\ref{cl:cut-tree} and Theorem~\ref{cor:T_sigma}. It can 
also be shown within the framework of Matroid Theory.
\end{proof}
%%%%%%%%%%%%%%%%%

The next two results address finer issues related to hypercuts,
in particular for $d=2$. First, we provide a characterization of 2-hypercuts 
(vs. general 2-coboundaries) in terms of their links, i.e., in purely 
graph-theoretic terms. 

Let $G=(V,E)$ be a graph. Call two adjacent edges $(u,v),\,(u,w) \in E(G)$ 
{\em V-equivalent} if $(v,w) \not\in E(G)$. 
I.e., the restriction of $G$ to $\{u,v,w\}$ is a "V" with $u$ at the apex.
Taking the transitive closure of this relation, we call
$G$ {\em V-connected} if any two edges of $G$ are V-equivalent.
\begin{theorem}\label{th:2-link}
Let $B$ be a 2-coboundary, and let $G=\Link_v(B)$ be its link with respect
to an arbitrary vertex $v$. Then, $B$ is a 2-hypercut iff $G$ is V-connected.
\end{theorem}
\begin{proof} 
Let $x$ be a vector with coordinates indexed by the edges of $K_n$. 
Consider the following system of equations in $x$. For each $e$ containing 
the vertex $v$, $x_e=0$; for each triangle $\sigma \not\in B$, 
$\sum_{e\in \sigma} x_e = 0$. We claim that this system of equations 
has a unique nontrivial solution iff $B$ is a hypercut. Indeed, by
definition, $x = 1_{E(G)}$ is one nontrivial solution, 
as $1_{E(G)}$ induces $B$. The existence of another nontrivial 
solution $x'$ is equivalent to existence of a nontrivial 2-coboundary $B'$ 
(induced by $x'$) strictly contained in $B$, as on every triangle
$\sigma \in \bar{B}$, $x'$ must sum to $0$. 
Recall that different links define different coboundaries.   
 
Assigning the forced value $0$ to all $x_e$ where $e$ contains $v$, and to
all $x_{(a,b)}$ where the triangle $\{a,b,v\} \not\in B$, we arrive at the
equivalent system of equations $x_{(a,b)} + x_{(b,c)} = 0$ whenever 
$a,b,c \in V- \{v\}$, and $(a,b)\,,\,(b,c) \in E(G); (a,c)\not\in E(G)$.
Thus, the edges in the same V-equivalence class must be assigned the same
value, but there is not restrictions for edges in different V-equivalence 
classes. We conclude that there is a unique solution iff there is one
V-equivalence class, i.e., $G$ is V-connected.
\end{proof}

Let us comment that a random graph $G$ on $n-1$ vertices is almost 
surely V-connected. (This is an easy exercise and we leave it to the reader.) 
Thus, in view of the above theorem, there are $2^{\Theta(n^2)}$ 
different 2-hypercuts.

Another comment is of a more geometrical nature. A closer look at 
the structure of 2-hypercuts $C$ reveals that not only for every
two different $\sigma,\tau \in C$ there exists a cycle $Z$ with 
$Z\cap C = \{\sigma,\tau\}$, but, moreover, $Z$ can be taken as a 
triangulation of the 2-sphere. This can be shown using the V-connectedness
of the links of $C$, first for $\sigma,\tau$ that share a common vertex, and then,
using transitivity, for any $\sigma,\tau$. This observation will not be
used in the rest of this paper.

How large/small can a $d$-hypercut be? A partial answer is provided by the 
following claim. 
%%%
\begin{claim}
The size of the minimum (nonempty) $d$-hypercut in $K_n^{(d)}$ is $n-d$.
The size of the maximum $2$-hypercut is ${n \choose 3} - O(n^2)$.
\end{claim}
\begin{proof} 
We start with the first statement, and prove it by induction on $n,d$. 
Since the minimum coboundary is a hypercut, it suffices to prove it
for coboundaries. The statement clearly holds for $d=1$ and for $n=d+1$. Assume 
that  the statement is true for all pairs $(n',d')$ where $n'< n, \; d'\leq d$.
Let $C$ be a nonempty $d$-coboundary, and let $v$ be a vertex. Consider $\Link_v(C)$.
Then, $|C| = |C'| + |\Link_v(C)|$, where $C'$ is the restriction of $C$ on
$V - \{v\}$, clearly a  $d$-coboundary of $K_{n-1}^{(d)}$. Recall that $\Link_v(C)$ 
cannot be empty. If $C'\neq \emptyset$, then by inductive hypothesis
$|C| \;\geq\; (n-1-d) + 1 \;=\; n-d$. Otherwise, by the previous discussion, 
$\Link_v(C)$ must be a $(d-1)$-coboundary of $K_{n-1}^{(d-1)}$, and thus 
by inductive hypothesis $|C| \;=\;  |\Link_v(C)| \;\geq\; (n-1)-(d-1) \;=\; n-d$. 
The bound is tight, as shown by a $d$-hypercut that consists all the 
$d$-simplices containing a fixed $(d-1)$-simplex $\tau$.

Let us just mention here without further elaboration that an alternative proof of 
the first statement can be obtained using the tools from the theory of simplicial 
matroids (see, e.g.,~\cite{CorLin} for a survey of this theory.)

For the second statement, consider the 2-coboundary $B$ of $K_n^{(2)}$whose link is a complete 
graph on $n-1$ points excluding a Hamiltonian cycle. It is easy to verify that the criterion 
of Theorem~\ref{th:2-link} holds, and thus $B$ is a 2-hypercut. A simple calculation 
shows that for $n \geq 5$, $|B| = {n \choose 3} - (n-1)(n-4)$.   
\end{proof}

We conclude this section with a result about the distribution of the sizes of $d$-hypercuts
in $K_n^{(d)}$, in particular when $d=2$. It should be noted that a similar but weaker 
result was shown earlier in~\cite{linial-meshulam} employing a somewhat more involved argument.
%%%%
\begin{theorem}\label{thm:cut-size1}
The number of $d$-hypercuts of size $\alpha n$ is at most $n^{c_d \cdot\alpha}$
where $c_d$ can be (very roughly) upper-bounded by $d(d+1)$. For $d=2$ we show 
a better upper bound of $(4n)^{3\alpha +1 } $. 
%We note that this enumeration can be slightly improved using better
%counting and the finer structure of $G$ as above, but this will have
%no consequence on the asymptotic of our following results.
\end{theorem}
\begin{proof}
Since $|C|=\alpha n$, the
average size of $|\Link_v(C)|$ is $(d+1)\alpha$, and therefore there
exists a vertex $v$ such that $|\Link_v(C)| \leq (d+1)\alpha$. Thus,
$|C|$ is induced by $G$ of size at most $(d+1)\alpha$. However,
setting $m={n \choose d}$, the number of such $G$'s is at most 
${m   \choose (d+1)\alpha} = O(n^{ d(d+1)\alpha})$. For $d=2$ we know that
$G$ is V-connected, hence it has at most one non trivial component
containing at most $3 \alpha$ edges and $3\alpha +1$
vertices.  Thus, the number of such $G$'s is at most
\[
{n \choose {3\alpha +1}} {{3\alpha +1 \choose 2} \choose {3\alpha}} ~\leq~ 
\left( \frac{{en}}{3\alpha +1} \right)^{3\alpha +1} 
\cdot \left( {{e\cdot 3\alpha (3\alpha +1)} \over {2\cdot 3\alpha}}\right)^{3\alpha} 
~\leq~ (4n)^{3\alpha +1} 
\]
\end{proof}

%%%%%%%%%%%%%%%%%%%%%%%%%%
\subsection{Geometrical Hypercuts}
%%%%%%%%%%%%%%%%%%%%%%%%%%%
{\em Geometrical} hypercuts are a very special subfamily of the more
general combinatorial hypercuts. They can be regarded as a different  
generalization of graph cuts to higher dimensions. Their definition
is quite intuitive, but it takes some effort to show that they are 
indeed hypercuts. As we shall see, they are particularly useful 
in dealing with Euclidean realizations of simplicial complexes. 
%%%%%%%%%
\begin{definition}\label{def:gcuts}
Let $\phi: V \mapsto \S^{d-1}$, the unit sphere of dimension $d-1$,
such that the points in the image are in a general position. 
The geometric hypercut $C$ is defined as the set of $d$-simplices whose 
image under $\phi$ contains the origin.
\end{definition}
%%%%%%%%%
%%
\begin{theorem}\label{th:gcuts-scuts}
Every geometrical $d$-hypercut $C$ is a combinatorial $d$-hypercut.
\end{theorem}
\begin{proof} 
We start with showing that for any $\sigma_1, \sigma_2 \in C$ it
holds that $\sigma_1 \sim \sigma_2~\mod\;\overline{C}$. Assume first
that the two simplices are disjoint. We use the following
cylindric construction. Consider two parallel copies of $\R^d$  in $\R^{d+1}$, 
each containing $\S^{d-1}$ with the $\phi$-image of $V$.
Choose $\sigma_1$ from first copy, and $\sigma_2$
from the second copy. Then, by the general position argument, the
boundary of the $\,{\rm conv} (\sigma_1 \cup \sigma_2) \subset \R^{d+1}$ is 
triangulated by $d$-simplexes. For every  $d$-simplex in this triangulation, 
consider the corresponding abstract simplex in $K_n^{(d)}$.  An easy projection
argument implies that all the simplices resulting from the  lateral
$d$-simplices in the above  triangulation (i.e., all but $\sigma_1$ 
and $\sigma_2$) are in $\overline{C}$. Since the union of all the
$d$-simplices in the above  triangulation forms a cycle (even over $\Z$), the 
statement follows. If the two simplices $\sigma_1$ and $\sigma_2$ are not
disjoint, we make the two copies of $\R^d$ intersect, such that all the 
common vertices (and only them) lie in the intersection, and proceed
in same manner.

We next argue that no $\sigma \in C$ is null homologous relatively to
$\overline{C}$. 
%It suffices to show that Claim~\ref{cl:cut-cyc} holds for $C$. 
Assume to the contrary that there exists a $d$-cycle $Z$ 
such that $Z\cap C$ contains a single simplex $\sigma$ containing the origin.
Using the central projection, we conclude that the realization of 
$\partial \sigma \;=\; \partial (Z - \sigma)$ is a retract of the realization of $Z$.
This can be refuted using standard basic algebraic topology arguments,
e.g., Sperner Lemma. Although classically the Sperner Lemma is used in
a weaker setting, it can be easily modified to apply here.
In addition to the classical argument, one needs only to notice that
since $Z$ is a cycle over $\Z_2$, the colored sub-simplices lying in
the abstract $(d-1)$-subsimplices of $Z$ (with the exception of $\partial
\sigma$),  appear even number of times in the Sperner sum, and thus 
contribute nothing.
\end{proof}

An important property of geometric $d$-hypercuts is
that the size of an intersection of such a $d$-hypercut with
a $d$-cycle $Z$ that is a boundary of $(d+1)$-simplex is
either 0 or 2. For combinatorial $d$-hypercuts this number
can be any even value between $0$ and $(d+2)$. While this property
does not characterize geometrical $d$-hypercuts, at least for $d=2$ it comes 
close (see~\cite{frankl-furedi-3graphs}). Moreover, using this property
and the discussion following Claim~\ref{cl:cut-cyc}, one gets another, 
less geometrical, proof of Theorem~\ref{th:gcuts-scuts}.

Only a tiny portion of combinatorial hypercuts are geometric. 
E.g., for $d=2$, the number of $d$-hypercuts is $2^{\Theta(n^2)}$,
as observed above, while the number of geometrical $d$-hypercuts can be shown to be
$2^{\Theta(n\log n)}$. This is the number of distinct (with respect to
the induced geometrical cuts) possible configurations of $n$ points on the
cycle. 

We conclude this section by mentioning a special subfamily of the
the geometric hypercuts, which also was suggested as a reasonable generalization
of the graph cuts. {\em Partition} $d$-hypercuts, studied e.g., 
in~\cite{lovasz,parekh-forestation},  correspond to partitions
${\mathcal P} = \{ V_1, \ldots, V_{d+1}\}$ of 
$V$ into $(d+1)$ disjoint nonempty parts. The hypercut $C_{\cal P}$ is defined
as $C_{\mathcal P} = \{\sigma \in K^{(d)}_n \;|\; |\sigma \cap V_i|=1,
i=1,2,\ldots,d+1\,\}$. It is easily verify that $C_{\cal P}$ is a geometrical
hypercut, and thus a hypercut. 

The following problem of Graham pertaining to the partition hypercuts
reflects the history of the early attempts at the proper definition of hypertrees, 
hypercuts etc. Graham defines a {\em $d$-forest} $F_d \subseteq K_n^{(d)}$ as a 
collection of $d$-simplices, such that for every $\sigma \in F_d$ there 
exists a partition hypercut $C$ such that $F_d \cap C = \sigma$. The problem
was to estimate the maximum possible size of a $d$-forest. It was
solved by Lov\'asz~\cite{lovasz,parekh-forestation} by introducing new  
(at the time) algebraic methods. 

Observe that the theory we have discussed so far allows to solve
Graham's problem in a rather obvious manner. 
Claim~\ref{cl:cut-cyc} implies that $F_d$ is acyclic, hence, by the
discussion in Section~\ref{sec:standard}, $|F_d| \leq {n-1 \choose
  d}$.  The tightness of the bound is witnessed by the $d$-hypertree
containing all the $d$-simplices that contain a fixed vertex $v \in
V$.
%%%%%%%%%%%%%%%%%%%%%%%%%%%%%%%%%%%%%%%%%%%%%%%%%%%%%%%%%%%%%%%%%%%%%%%%%%%%%
%
%
\section{Abstract Volumes}
\label{sec:vol}
\subsection{Basic Notions} 
Let $K_n^{(\leq d)}$ be the simplicial complex on the underlying set $V$ of 
size $n$ containing all the simplices of dimension $\leq d$ on $V$.
We define the abstract $d$-dimensional volume function $\v^{(d)}: K_n^{(\leq d)}
\mapsto \R^+$ as a real nonnegative function with the following properties:
(*) the simplices of dimension $< d$ have value $0$;  (**)  the values of $d$-simplices 
satisfy the following  generalization of the triangle inequality: 
\begin{equation}
\label{tri-ineq}
\mbox{\em 
For every $d$-cycle $Z$ of $K_n^{(d)}$, and every $\sigma \in Z$, it
holds that $\sum_{\sigma' \in Z - \sigma} \v^{(d)}(\sigma')
\;\geq\; \v^{(d)}(\sigma)\;$.
}
\end{equation}
It is easy to verify that for $d>1$ the condition (**) cannot be replaced 
by a requirement on cycles of bounded size.  

The most natural example of the volume function
is the {\em Euclidean} volume: given an embedding $\phi$ of $V$ into  
an Euclidean space,  the volume of  a $d$-simplex $\sigma$,
is the Euclidean $d$-volume of $\conv(\phi(\sigma))$.

Another important example is the analog of the shortest-path metric.
Let $X\subseteq K_n^{(d)}$ be a connected (i.e., containing a $d$-hypertree) 
subcomplex with nonnegative weights on its $d$-simplices.  
The volume $\v_X$ induced by $X$  on $K_n^{(d)}$ is defined by  
$\v_X = \min_{D_\sigma \subseteq X} \sum_{\sigma' \in D_\sigma} w_{\sigma'}$,
where $D_\sigma$ is a $\sigma$-{\em cap}, i.e.,  $\sigma \cup D_\sigma$ is a cycle.
(In  particular, $\sigma$ itself is $\sigma$-cap.)

The last example are {\em cut volumes}, which play a central role in this
paper. Let $C$ be a $d$-hypercut in $K_n^{(d)}$.
The corresponding volume function $\v_C^{(d)}$ assigns 1 to every $\sigma \in
C$, and 0 to every $\sigma \not\in C$. To see that a cut volume is indeed a volume,
it suffices to notice that a 0/1 function on $d$-simplices may fail to be a volume
function iff there exists a cycle $Z$ were all but one $\sigma \in Z$
have value $0$. By Claim~\ref{cl:cut-cyc}, such $Z$ does not exist 
for $\v_C^{(d)}$. 
 
Volume functions on $V$ are closed under addition and multiplication by a 
constant, and thus form a cone in $\R_{+}^{n \choose {d+1}}$.  The extremal 
volumes in this cone are, as always, of particular interest. The following 
theorem provides a full characterization of 0/1 extremal volumes. Perhaps more
important, it also establishes their inapproximability but any other metric. 

The {\em multiplicative distortion} between two $d$-volume functions 
$\v_1$ and $\v_2$ on $V$ is defined similarly to the metric distortion, i.e.,
\[
\dist(\v_1,\v_2) \;=\; \max_\sigma {{\v_1(\sigma)} \over {\v_2(\sigma)}} \cdot
\max_\sigma {{\v_2(\sigma)} \over {\v_1(\sigma)}}\,.
\]
\begin{theorem}
\label{thm:decomp}
A 0/1 volume function $\v^{(d)}$ is extremal iff it is a cut volume.
Moreover, the distortion between such $\v^{(d)}$ and any other volume 
function $\v_1^{(d)}$ is infinite unless $\v_1^{(d)} = \alpha \cdot \v^{(d)}$
for some positive constant $\alpha$.
\end{theorem}
\begin{proof}
Let $\v^{(d)}$ be a cut $d$-volume function defined by a hypercut $C$.
Assume that $\v^{(d)} \;=\; \v_1^{(d)} + \v_2^{(d)}$. Consider
$\v_1^{(d)}$. It must be $0$ outside of $C$. Since any two
$\sigma,\sigma' \in C$ satisfy $\sigma \sim \sigma'\,\mod\,\overline{C}$,
there exists a cycle $Z=Z_{\sigma,\sigma'}$ such that 
$Z\cap C \;=\; \{\sigma, \sigma'\}$. Since all
the $d$-simplices in $\overline{C}$ have volume $0$, the generalized
triangle inequality implies that $\v_1^{(d)}(\sigma) \;=\; \v_1^{(d)}(\sigma')$. 
Thus, $\v_1^{(d)} \;=\; \alpha \cdot \v^{(d)}$, as claimed.

For the other direction, consider an extremal $0/1$ $d$-volume 
function $\v^{(d)}$. Define $C\subset K_n^{(d)}$ as $
C \;=\; \{\sigma \;|\;  \v^{(d)}(\sigma)=1 \}$. 
Clearly, no $\sigma \in C$ is null homologous relatively to $\overline{C}$, since
otherwise the generalized triangle inequality would imply $\v^{(d)}(\sigma)=0$. 
Consider the equivalence relation $\sim$ on $C$, i.e., 
the homology $\mod\;\overline{C}$. It suffices to show that it contains a 
single equivalence class. Assume to the contrary that there is an 
equivalence class $C'$ strictly contained in $C$.
Define $\v_1^{(d)}$ and $\v_2^{(d)}$ as follows. Outside of $C$ both are
$0$. For $\sigma \in C \setminus C'$, $\v_1^{(d)}(\sigma)\;=\; \v_2^{(d)}(\sigma) 
\;=\; \frac{1}{2}$; for $\sigma \in C'$, $\v_1^{(d)}(\sigma)\;=\; 0.4$, 
and $\v_2^{(d)}(\sigma)\;=\; 0.6$. The definition of $C'$ implies that both  
$\v_1^{(d)}$ and $\v_2^{(d)}$ are volume functions, contradicting the
assumption that $\v^{(d)}$ is extremal.

The second statement follows easily along the same line of reasoning. The
support of any volume function approximating such $\v^{(d)}$ must coincide
with the support of $\v^{(d)}$, and moreover, arguing as above, it must be
constant on it.
\end{proof}

The above theorem provides an additional motivation to our definition of 
hypercuts, this time from volume theoretical perspective.

Much of the modern theory of finite metric spaces is devoted to the
study of special metric classes that constitute a sub-cone of the metric cone,
notably $\ell_1$ metrics and $NEG$-type metrics. Crucially for applications,
any metric on $n$ points can be approximated by a special metric with
a bounded distortion $c_n$. E.g., for  $\ell_1$ the rough bound of
$O(n)$ on distortion follows from the minimum spanning tree argument,
and the much better $O(\log n)$ bound is implied by Bourgain's
Theorem \cite{bourgain}. Theorem~\ref{thm:decomp} implies that any (closed)
sub-cone of volume functions with the approximation property {\em must} 
contain the cone spanned by the cut volumes. Moreover, as we shall soon see, 
this cone already has the required property. This justifies the following 
definition.

\begin{definition}\label{def:l1-vol}
Analogously to one dimensional case, we define $\ell_1$ $d$-volumes 
to be the nonnegative combinations of cut $d$-volumes. 
\end{definition}
Clearly, $\ell_1$ $d$-volumes constitute a sub-cone of $d$-volumes.
%%%%%%%%%%%%%%%%%%%%%%%%%%%%%%%%%%%%%%%%%%%%%
%% 
%%
%%%%%%%%%%%%%%%%%%%%%%%%%%%%%%%%%%%%%%%%%%%
\subsection{$\ell_1$ Volumes}
The most basic properties of $\ell_1$ metrics are
that they contain the class of tree-metrics and the 
class of Euclidean metrics. The situation with
$\ell_1$ $d$-volumes turns out to be fully analogous.

 \begin{theorem}\label{th:trees}
Let $T$ be a (spanning) $d$-hypertree with nonnegative weights on the 
$d$-simplices. Then, the induced $d$-volume function 
$\v_T^{(d)}$ is $\ell_1$.
\end{theorem}
\begin{proof}
Recall the definition of $C_{T,\sigma}$ from Theorem~\ref{cor:T_sigma}.
We claim that $\v_T^{(d)} \,=\, \sum_{\sigma \in T} \v_{C_{T,\sigma}^{(d)}}$. 
For $\tau \in T$ this follows at once, while for $\tau\not \in T$, 
$\sum_{\sigma \in S} \v_{C_{T,\sigma}}^{(d)}(\tau)$ is equal to
the sum of weights of all the $\sigma$'s in $S$ belonging to the cycle 
created by adding $\tau$ to $T$, as it should be.
\end{proof}

This implies the following approximability result.
 
\begin{theorem}\label{th:MST}
Any $d$-volume on $V$ can be approximated by an $\ell_1$ $d$-volume
with distortion at most ${{n-1} \choose {d}}$.
\end{theorem}
\begin{proof}
Let $\v^{(d)}$ be a $d$-volume function on $K_n^{(d)}$, and let $T$
be the minimum (spanning) hypertree with respect to $\v^{(d)}$.
Then, for $\sigma \in T$,  $\v_T^{(d)}(\sigma) \;=\;
\v^{(d)}(\sigma)$. For $\sigma \notin S$, much like the MST in
graphs, $\sigma$ must be the heaviest $d$-simplex in the cycle $|Z|$
created by adding $\sigma$ to $T$. Since the size of $Z$ is at most 
$1 + |T| \;\leq\; 1 + {{n-1} \choose d}$, the statement follows.
\end{proof}

While the upper bound on distortion of Theorem~\ref{th:MST} is
probably too rough and the true exponent of $n$ is probably smaller,
we shall see in what follows that even for $d=2$ the distortion can be
as large as $\Omega(n^{1\over 5})$. Thus, in general it is polynomial,
and not logarithmic as in the case for $d=1$ (Bourgain's
Theorem~\cite{bourgain}). Another important difference between $d=1$
and $d=2$ is that the Euclidean $2$-volumes, and in fact even their
nonnegative combinations, are unable to approximate {\em at all} even
%in: changed
the simplest $2$-volume functions, e.g., set $V=\{0,1,2,3,4\}$ and 
$vol(\{i, i+2, i+3\}) =1$, where $+$ is taken $\bmod~ 5$, and
$vol(\sigma)=0$ for any other $\sigma$. It is easy to see that this
function is a volume and in fact geometrical cut volume. However, any
geometrical realization that approximate it can not collide any two
points, which implies in turn, that it must assign a strictly positive
volume to a $\{i.i+1,,i+2\}$ simplex. 

Next we address the containment of Euclidean volumes in $\ell_1$-volumes. 
\begin{theorem}\label{cl:Eu-gcuts}
Any Euclidean $d$-volume is an $\ell_1$ $d$-volume. 
In fact, it is a nonnegative combination of geometrical hypercuts. 
\end{theorem}
\vspace{-0.3cm} \begin{proof}{\bf (Sketch)} 
The proof proceeds in three steps. First, observe that the random
projection of a finite dimensional Euclidean space on $\R^{d}$
preserves (in expectation) the $d$-volumes up to scaling. Thus,
it suffices to consider Euclidean $d$-volumes realizable in $\R^{d}$.
Next, observe that given an embedding of $V$ points in $R^d$,
the corresponding Euclidean volume function $\v^{(d)}$ satisfies
$\v^{(d)}\;=\; \int_{\R^d} \v_p^{(d)}$, where $\v_p^{(d)}(\sigma) = 1$
if the realization of $\sigma$ contains $p$, and $0$ otherwise.
Treating $p$ in $\v_p^{(d)}$ as the origin, one can realize the
same function by projectively mapping the points to $\S^{d-1}$,
which implies that $\v_p^{(d)}$ is geometrical cut volume.
Measure $0$ argument take care of the degeneracies. Finally,
by Theorem~\ref{th:gcuts-scuts}, every geometrical hypercut 
is a (combinatorial) hypercut, and thus we get an $\ell_1$ volume
with the same values as the original Euclidean volume.
\end{proof}

The main negative result of this section is the following lower bound on distortion
of approximating general 2-volumes by $\ell_1$ 2-volumes. 
%in: added
On the way
we define a $d$-dimensional analog of the graphical 'edge-expansion',
which is of independent interest.
%%%%%%%%%%%%%%%%
\begin{theorem}\label{thm:main-distortion}
There exists a 2-volume function  such that any $\ell_1$
volume distorts it by at least $\tilde{\Omega}(n^{1/5})$.
\end{theorem}
%%%%%%%%%%%%%%%%%
Let us first outline the proof. Using the methods originally developed 
for the one-dimensional case, we construct a connected 2-dimensional 
simplicial complex $K$ with unit weights on its 2-simplices, such that 
on one hand is has a constant {\em normalized expansion}, 
and on the other hand $\v_K$ has large average value. The existence
of such $K$ implies that distortion of embedding $\v_K$ into
$\ell_1$ is large. Formally, given a $K$ as above,
consider the following Poincare-type form over the 2-volumes:
\begin{equation}\label{eq:form}
  F_K(\v) = \frac{\sum_{\sigma \in K} \v(\sigma)}{\av(\v)}\,,
\end{equation}
where $\av(\v) \;=\; 
\frac{1}{{n \choose 3}} \cdot \sum_{\sigma \in \K} \v(\sigma)$. 
By a standard argument frequently used in the theory of metric spaces,
the distortion of embedding $\v_K$ into $\ell_1$ is lower-bounded by
\begin{equation}\label{eq:form1}
  \dist(\v_K \hookrightarrow \ell_1) \;\geq\; \frac{\min_{vol \in \ell_1} 
{F_K( vol)}}{F_K(\v_K)}\,.  
\end{equation}
Keeping in mind that $K$ is unit-weighted, and that any $\v \in \ell_1$
is a nonnegative combination of cut-volumes, we conclude that the above
minimum necessarily occurs on cut-volume, and thus Eq.~\ref{eq:form}
becomes:
\begin{equation}\label{eq:form2}
  \dist(\v_K \hookrightarrow \ell_1) \;\geq\; \av(\v_K) \cdot 
\min_{C:~\mbox{2-hypercut}} \frac{|K \cap
   C|/|C|}{|K|/ {n \choose 3}} 
\end{equation}
Observe that for a graph $G$ the analogous expression 
\[
\min_{C=E(A,\overline{A}):~\mbox{cut}} \frac{|E(G) \cap C|/|C|}{|E(G)|/{{n \choose 2}}} 
\;=\; \min_{A \subset V, |A| \leq n/2} \left\{
{{|E(A,\overline{A})|}\over |A|} \cdot 
{1 \over {\mbox{average degree of}~ G}} \right\}
\cdot {{n-1}\over{n - |A|}}\,,
\]
is the normalized edge expansion of $G$ up to a factor of 2. 
By analogy, we define 
\begin{definition}
Let the {\em normalized (face) expansion} of $K \subseteq \K$ be the 
value of
\[
\min_{C:~\mbox{2-hypercut}} \frac{|K \cap C|/|C|}{|K|/{{n \choose 3}}}\,. 
\]
\end{definition}
I.e., the normalized expansion of $K$ is the ratio between the minimum density
of $K$ with respect to a hypercut, and the density of $K$ with respect to $\K$.

Let $\K(n,p)$ be the 2-dimensional analog of 
the Erd\"os-R\'enyi $G(n,p)$, where $\sigma \in \K$ is selected with probability 
$p = 25\log n/n$ randomly and independently from the others. %We set $p = 96\log n/n$.
Theorem~\ref{thm:main-distortion} follows from the following two Lemmas.
\begin{lemma}\label{cl:123}
For $K \in K^{(2)}(n,p)$ as above, $\av(vol_K) \geq \tilde{\Omega}(n^{1/5})$ 
with probability $1-o(1)$.
\end{lemma}
\begin{lemma}\label{lm:expander}
The face expansion of $K \in K^{(2)}(n,p)$ is almost surely $\geq 0.5$.
\end{lemma}
Observe that  Lemma~\ref{lm:expander} implies that $K$ is connected, since if
all 2-hypercuts meet $K$, then by Corollary~\ref{thm:blocker} $K$ must contain a 
(spanning) 2-hypertree. Thus, it strengthens the main result of~\cite{linial-meshulam} 
at the price of getting worst constants. 
%We suspect, however, that with a little
%more care one could show that the thresholds for connectivity and expansion
%coincide, as they do for random graphs. 

Before starting with the proof of Lemma~\ref{cl:123}, let us first 
establish the following combinatorial result.
\begin{lemma}\label{lem:cycle_card}
Let $Z$ be a 2-dimensional cycle $Z$, then, \;$|V(Z)| \;\leq\; |Z|/2+2$.
\end{lemma}
\begin{proof}
%%in: I prefer here the combinatorial proof with a ref. to the
%%topolgical one. This topological proof is the ONLY non-basic proof here.
Clearly, $\Link_v(Z)$ is an Eulerian (1-dimensional)
graph. As long as there is a vertex $v \in V(Z)$ for
which $\Link_v(Z)$ is not a simple cycle, do the following. 
Let $A_1, \ldots, A_r$ be the decomposition of $\Link_v(Z)$  
into edge-disjoint cycles. We introduce a new copy of $v$, $v_i, i=1, \ldots r$ 
for each $A_i$, and replace each original 2-simplex $\{v,x,y\}$ containing $v$ with
a new 2-simplex $\{v_i,x,y\}$ where $(x,y) \in A_i$. This
yields a new simple cycle $Z'$. Carry on with the this process
on $Z'$ etc. Since each time we produce a new $2$-cycle with the same number of faces,
but less vertices whose link is not a simple cycle, the process must terminate with  a 
$2$-cycle $Z^*$ with all links being simple cycles. Such $Z^*$, using the 
language of algebraic topology, is a (vertex-) disjoint union of triangulations of 2-dimensional 
surfaces without boundary. Without loss of generality, assume that there is a single surface.
It is known~\cite{massey} that its Euler characteristics satisfies 
\begin{equation}\label{eq:119}
\chi(Z^*) \;=\; |V(Z^*)| - |E(Z^*)| + |Z^*| \;\leq\; 2
\end{equation} 
Observe that every edge $e$ in $Z^*$ appears in exactly two faces,
and thus $2|E(Z^*)| = 3 |Z^*|$. Plugging this into Equation (\ref{eq:119})
implies the Lemma for $|V(Z^*)|$, and hence for $|V(Z)|$. We note that
while this proof uses Equation (\ref{eq:119}), which is non-trivial
and outside of this context, there is also an elementary proof using
reduction to smaller $n$'s.
\end{proof}

Next, we address Lemma~\ref{cl:123}. 

\begin{proof} {\bf (of Lemma~\ref{cl:123})}
By Markov inequality $K$ almost surely contains $o(n^3)$ 2-simplices,
and thus $\av(\v_K)$ is determined by the 2-simplices $\sigma \not \in K$. 
For each such $\sigma$, $\v_K(\sigma)$ is the size of the smallest 
$K$-cap of $\sigma$, i.e., the minimum subset of simplices in $K$ that
together with $\sigma$ form a simple cycle.
Let us denote this cap by $\Cp_K(\sigma)$.
Thus, to show that $\av(\v_K)\geq \Omega(\lambda)$
(w.h.p.), it suffices to argue that the number of $\sigma \notin K$ for which the corresponding 
$\Cp_K(\sigma)$ has size  less than $\lambda$, is $o(n^3)$ (w.h.p). 
Let $N_\lambda$ be this number. 
%
%To compute $\beta(\lambda)$ we may count all cycles $Z$ of size
%exactly $k$, for all $k \leq \lambda$, for which exactly one face is
%missing from $K$, with multiplicity $k$ (as any of its members can
%serve as $\sigma$). 
Let $n_k$ be the number of simple cycles of size exactly
$k$ in $\K$. Then, 
  \begin{equation}
    \label{eq:no-cycles}
    E[N_\lambda] \;=\; \sum_{k=4}^\lambda k \cdot n_k \cdot p^{k-1} (1-p)
  \end{equation}
Now, by Lemma \ref{lem:cycle_card}, a cycle of size $k$ has at most
$k/2 +2 $ vertices. Fixing $t =  k/2 + 2$ vertices, the number
of size-$k$ cycles on these vertices is clearly bounded by $t^{3k}$. 
Hence $n_k \leq (k/2 + 2)^{3k} \cdot {n \choose (k/2 + 2)} \;\leq\; n^2 \cdot (k^{2.5} \sqrt{n})^{k}$.
Plugging this bound on $n_k$, and the value of $p$ 
into Equation (\ref{eq:no-cycles}), we get,
$$E[N_\lambda] \;\leq\; n^2 \sum_{k=4}^{\lambda} (k^{2.5} \cdot
  \sqrt{n})^k \cdot k \cdot \left(\frac{25 \log n}{n}\right)^{k-1} \;\leq\;
  \frac{n^3}{25 \log n}  \cdot \sum_{k=4}^{\lambda} k \left(\frac{k^{2.5} \cdot
  25\log n}{\sqrt{n}}\right)^k  $$ 
Choosing $\lambda = \frac{n^{1/5}}{50\log n}$, we conclude that 
$E[N_\lambda] = O(n \log^3 n) = \tilde{O}(n)$, and by the Markov inequality we are done.
\end{proof}%cl:123
\\ 
%Next, we turn to Lemma~\ref{lm:expander}, the Expansion Lemma.
\begin{proof} {\bf (of Lemma~\ref{lm:expander})}
For a hypercut $C$, let $\gamma_K(C) = \frac{|K \cap C|/|C|}{|K|/{n \choose 3}}$. 
We shall first estimate the probability that
$\gamma_K(C)<0.5$ for any {\em fixed} hypercut $C$, and then use
the union bound to conclude that almost surely no such
hypercut exists.

Observe first that $|K|$ is almost surely tightly concentrated around its
mean which is $E[K] = p \cdot {n \choose 3}$. Thus instead of discussing 
$\frac{|K \cap C|/|C|}{|K|/{n \choose 3}}$,
we may safely discuss $\frac{|K \cap C|/|C|}{|E[K]|/{n \choose 3}} =
\frac{|K \cap C|}{p \cdot |C|}$. 
Next, observe that $|K \cap C|$ is a sum of $|C|$ i.i.d Bernuli 
variables, and its expectation is precisely $p|C|$. %|E[K]|\cdot {|C|/{n \choose 3}}$.   
Thus, by Chernoff bound,
\[
\Pr\left(\gamma_K(C) < 0.5\right) \;=\; \Pr\left(|K \cap C| < p \cdot |C|/2\right) \;\leq\;
e^{-p \cdot |C|/8}\,. 
\]
Let $m_s$ be the number of 2-hypercuts of size $s$ in $\K$. By 
Theorem~\ref{thm:cut-size1}, $m_s \leq  (4n)^{1+3s/n}$.
Thus, the union bound implies that the probability that a bad $C$ exists is at most
\[
\sum_{s\geq n-2} m_s \cdot e^{-p \cdot s/8} \;\leq\; 
4n \sum_{s\geq n-2} 
e^{\left(-{25 \over 8} {{\log n} \over n}  + 
{{3\log (4n)} \over n}\right)\cdot s} \;=\;  o(1)\,.
\]
\end{proof}
%%%%%%%%%%%%%%%%%%%%%%%
%
%   NEWER
%
%%%%%%%%%%%%%%%%%%%%%%%
\subsection{Geometrical $\ell_1$ Volumes, Exact and Negative Type Function}
\label{sec:GEN}
By geometrical $\ell_1$ volumes we mean nonnegative sums of geometrical cut volumes.
As implied by Theorem~\ref{cl:Eu-gcuts}, Euclidean volumes belong to this class.
The following examples show that geometrical $\ell_1$ volumes capture 
other geometrically defined volume functions as well.
\\ \\
{\bf Example 1.} {\em Let $f$ be a nonnegative weighting of $(d-1)$-simplices of 
$K_n^{(\leq d)}$. Define a $d$-volume function $\v^{(d)}$ on $K_n^{(d)}$ by 
\[
\v^{(d)}(\sigma) ~=~ \sum_{\mbox{$(d-1)$-simplex $\tau \subset \sigma$}} f(\tau)\,.
\]
Then, $\v^{(d)}$ is a geometrical $\ell_1$ volume since it can be represented by 
$\v^{(d)}= \sum_{C_\tau, \tau ~\mbox{is} (d-1)-\mbox{simplex}} f(\tau) \cdot \v^{(d)}_{\tau}$, where  
$\v^{(d)}_{\tau}$ is a (geometrical) cut volume assigning 1 to the $d$-simplices 
containing $\tau$, and 0 to the rest. In particular, the Euclidean perimeter, surface 
area, etc., are geometric $\ell_1$ $d$-volumes. 
}
\\ \\
{\bf Example 2.} {\em Let $\cal H$ be a family of $n$ affine hyperplanes in 
general position in $\R^d$, indexed by $[n]$. Assign to every $d$-simplex of $K_n^{(d)}$ the Euclidean volume of the 
unique bounded cell of $\R^d$ formed by the corresponding $(d+1)$ hyperplanes. 
The resulting $d$-volume function (which can be interpreted as a measure of 
disagreement between the $(d+1)$-tuples of hyperplanes) is geometrical $\ell_1$. 

The proof is quite similar to that of Theorem~\ref{cl:Eu-gcuts}.  It
suffices to show that for each $p\in \R^d$, the set of $d$-simplices
$\sigma$ corresponding to the $(d+1)$ tuples of hyperplanes containing
$p$ in their bounded cell, is a (geometric) hypercut. Indeed, map each
hyperplane $h$ to $p_h\in \R^d$, the basis of the perpendicular from
$p$ to $h$. Clearly, $p$ is contained in the bounded cell of some
$(d+1)$ hyperplanes $\{h\}$ iff $p$ belongs to the geometrical simplex
$\{p_h\}$. The conclusion follows.  }
\\ \\
While so far our basic notions (i.e., boundary operator, cycles, and
coboundaries) were over $\Z_2$, in the context of the geometric
$\ell_1$ volumes it will be helpful to (shortly) discuss the
corresponding theory over $\R$. The presentation is not going to be
entirely self contained, and we refer the reader to the first
chapters of~\cite{munkres} for the background.

As before, we consider ${n \choose {d}} \times {n \choose {d+1}}$ incidence matrix $M_d$ 
over the reals, whose rows are indexed by (arbitrarily oriented) $(d-1)$-simplices, 
and the columns 
are indexed by (arbitrarily oriented) $d$-simplices. This time, $M_d(\tau,\sigma) = 1$ 
if $\tau \subset \sigma$ and its orientation is consistent with the orientation induced 
by $\sigma$ on its boundary, $M_d(\tau,\sigma) = -1$ if $\tau \subset \sigma$ but 
the orientations are inconsistent, and $M_d(\tau,\sigma) = 0$ if 
$\tau \not\subset \sigma$.

The boundary operator $\partial: K_n^{(d)} \mapsto K_n^{(d-1)}$ is defined by 
$M_d 1_\sigma = 1_{\partial \sigma}$, and is linearly extended to act on formal sums
of $d$-simplices with real coefficients. A $d$-coboundary 
$B \in \R^{n \choose {d+1}}$ (i.e., a real function on $d$-simplices) is a 
vector in the left image of $M_d$. That is, $B^T=x^T M_d$ for some 
$x \in \R^{n \choose {d}}$. 
%%%%%%%%%%%%%%%%%%%%%%%%%%%%%%%%%%%%%%%%%%%%%%%%%%%%%%%%%%%
% One way of thinking about the real $d$-coboundaries is as follows. Consider a realization
% of $K_n^{(d)}$ in $\R^d$, i.e., a mapping $\phi:V \mapsto \R^d$ which is extended to all simplices
% $\sigma$ over $V$ by taking the convex hull of the image of $\sigma$'s vertices (with a suitable
% orientation). Let $f$ be an {\em exact} real valued function on $\phi(K_n^{(d))$. 
% Then, $b$ as a real function on $d$-simplices defined by $b(\sigma) = \int_{\phi(\sigma)} f$,
% is a real $d$-coboundary.
%%%%%%%%%%%%%%%%%%%%%%%%%%%%%%%%%%%%%%%%%%%%%%%%%%%%%%%%%%%%

An equivalent definition of a real $d$-coboundary, based on the fact that
$H^{d-1}(K_n^{(d)},\R)=0$, is:~ $B\in \R^{n \choose {d+1}}$ is a real $d$-coboundary 
iff it sums up to 0 on the boundary of any $(d+1)$-simplex. I.e.,
$B^T M_{d+1} = 0$.  

\begin{definition}
A real nonnegative function $F:K_n^{(d)} \mapsto \R_+^d$ is {\em exact} 
if it is an (entrywise) absolute value of a real $d$-coboundary of $K_n^{(d)}$. \par
A real nonnegative function $T:K_n^{(d)} \mapsto \R_+^d$ is of {\em negative type}
if it is a sum of (entrywise) squares of real $d$-coboundaries of 
$K_n^{(d)}$.
\end{definition}
The exact $d$-volumes can be viewed as a $d$-dimensional analog of line metrics.
Observe that exactness does not depend on the orientations used in the definition
of $M_d$. Observe also that in the alternative theory where the generalized triangle inequality 
of Eq.~(\ref{tri-ineq}) is required to hold only for orientable cycles (i.e., cycles over $\R$),  
the exact function are $d$-volumes. However, in some important case they are $d$-volumes 
according our original definition as well: 
\begin{theorem}
\label{th: NEG}
Cut volumes corresponding to geometrical $d$-hypercuts are exact. So are 
the Euclidean $d$-volumes realizable in $R^d$. Consequently, geometrical
$\ell_1$ $d$-volumes, as well as the sums of squares of Euclidean $d$-volumes, 
are of negative type.
\end{theorem}  
\begin{proof} {\bf (Sketch)} Consider a realization of $K_n^{(d)}$ in $\R^d$ defining
the geometrical hypercut $C$, or the $d$-Euclidean volume under the consideration, 
with all $d$-simplices oriented in the same manner. I.e., left to right for $d=1$,
clockwise for $d=2$, etc. 
Observe that 
the origin is contained in either zero or two $d$-simplices belonging
to the boundary of any 
$(d+1)$-simplex $\zeta$. In the latter case 
one of these simplices is necessarily oriented in a manner consistent with the 
orientation induced by $\zeta$, and the other is not. Hence, 
$(\v_C^{(d)})^T M_{d+1}=0$, and thus  $\v_C^{(d)}$ is exact. The Euclidean volume,
which is the integral of geometrical cut volumes defined by all $p\in \R^d$
with respect to a fixed realization of $K_n^{(d)}$, must also be exact by a linearity
argument. 

The second statement directly follows from the first for geometrical $\ell_1$ $d$-volumes,
as cut volumes take values 0/1. For general Euclidean volumes, recall that the
square of any Euclidean $d$-volume (no matter in what dimension it is realized) is the
sum of squares of its projections on all subsets of $d$ coordinates. I.e., it is a sum 
of squares of $d$-Euclidean volumes.
\end{proof}

To conclude this section, observe that Theorem~\ref{th: NEG} provides an alternative
proof of Theorem~\ref{th:gcuts-scuts}, and in fact a bit more: a geometrical hypercut
intersects not only every $\Z_2$-hypertree, but also any $\R$-hypertree. Indeed, any
$d$-coboundary of $K_n^{(d)}$, in particular an appropriately signed $v_C^{(d)}$, 
that takes value $0$ on a basis of the space of columns of $M_d$ 
(i.e., on a $\R$-hypertree), must be identically $0$ on $K_n^{(d)}$, contrary to the 
definition of $v_C^{(d)}$.  
%%%%%%%%%%%%%%%%%%%%%%%%%%%%%
%
%%%%%%%%%%%%%%%%%%%%%%%%%%%%
\subsection{Dimension Reduction for $\ell_1$ Metrics and Volumes}\label{sec:reduc}
Given an $\ell_1$ $d$-volume $\v = \sum_{C \in {\mathcal C}} \lambda_C
\cdot v_C$, where ${\mathcal C}$ is a collection of $d$-hypercuts, ${\rm v}_C$ is the
cut volume associated with $C$, and $\lambda_C$ are positive reals,
$|{\mathcal C}|$ is the {\em cut-dimension} of this particular
representation of $\v$. We define the {\em cut-dimension} of $\;\v\;$ as
the minimum possible cut-dimension of any representation of it. 

Let the {\em cut cone} be the convex cone formed by all $\ell_1$ $d$-volumes
on $K_n^{(d)}$. The extremal rays of this cone are the cut-$d$-volumes.
\begin{claim}
\label{cl:cutcone}
The cut cone has full dimension. 
\end{claim}
\begin{proof}
Assume that a function $f:K_n^{(d)} \mapsto \R$ sums up to 0 on every 
hypercut (and therefore, by Theorem~\ref{cl:cut-boundary}, on any $d$-coboundary 
of $K_n^{(d)}$). It suffices to show that $f$ is identically 0.   
Let $\sigma$ be any $d$-simplex in $K_n^{(d)}$, and let $\tau_1,
\tau_2$ be distinct 
$(d-1)$-dimensional faces of $\sigma$. Let $B_1,B_2$ and $B_{12}$ be the $d$-coboundaries 
in $K_n^{(d)}$ induced by $\tau_1$, $\tau_2$ and $\{\tau_1,\tau_2\}$ respectively. 
Then, $0\;=\; f(B_1) + f(B_2) - f(B_{12}) \;=\; 2f(\sigma)$, and the claim follows.
\end{proof}

Since the cut cone is a subset of $\R^{{n \choose {d+1}}}$, Caratheodory Theorem
implies that the cut-dimension of any $\v^d$ is at most ${{n \choose {d+1}}}$.
Moreover, since the cut cone has a full dimension, all but a 
0-measure subset of $\ell_1$ $d$-volumes have precisely this cut-dimension.  

The dimension reduction phenomenon is the dramatical drop in 
the cut dimension when one is allowed to replace an $\ell_1$-volume $\v$ by an 
$\epsilon$-close $\ell_1$-volume  $\v'$. The proximity in our case is 
measured by the point-wise ratio between $\v$ and $\v'$, which should lie within 
$(1 \pm \epsilon)$. 
I.e., the multiplicative distortion between $\v$ and $\v'$ is 
$\leq \frac{1+\epsilon}{1-\epsilon}$.   

We show that the dimension reduction phenomenon occurs for  
$\ell_1$-volumes for any $d$. For $d=1$ and, more generally, for geometrical $d$-volumes
of any dimension, we refine the argument, and get
a better bound. In order to do this, we rely on some general sparsification tools
to be developed and discussed in detail in the next chapter. Here we present
only the statements of these results, and then proceed to apply them in our
setting.

The geometric formulation is as follows.  Let $\cal C$ be a family of 
nonnegative vectors in  $\R^m$, and let $\cone({\cal C})$ be the convex cone spanned by it.
The goal is, given a vector $w \in \cone({\cal C})$, to produce a small subset 
${\cal C}' \subset {\cal C}$ and a vector $w' \in \cone({\cal C}')$ that (pointwise) 
approximates $w$ up to a multiplicative factor of $1 \pm \epsilon$. 

The same can be conveniently reformulated in the matrix notation. 
Let $M$ be a $m\times |{\cal C}|$ real nonnegative matrix. Then, given a nonnegative
vector $\lambda \in \R^{|{\cal C}|}$, the goal is to produce a new 
$\lambda' \in \R^{|{\cal C}|}$ such that on one hand $w'=M\lambda'$ approximates  
$w=M\lambda$ up to a multiplicative factor of $1 \pm \epsilon$, and on the other hand
$\lambda'$ has small support. The columns of $M$ are the vectors of 
$\cal C$, and $\lambda,\lambda'$ are coefficients of nonnegative combinations 
of these vectors.

An upper bound on the size of support of $\lambda'$ will be given 
in terms of certain parameters of the matrix $M$ alone, not depending on $\lambda$. 
%We restrict the discussion here to Boolean matrices $M$.

\begin{definition}
The {\rm triangular rank} of a matrix $M$, $\,\trk(M)$, is the size of 
the largest lower-triangular square minor of $M$ with strictly positive
diagonal. The rows and the columns of the minor may appear in order 
different from that of $M$.  

The {\em square-root rank} of a nonnegative matrix $M$, $\,\srk(M)$, is the
the minimum possible rank (over $\R$) of a matrix $Q$ where 
$Q_{ij} = \pm \sqrt{M_{ij}}$. In particular, if $M$ is Boolean, then $Q$ ranges over 
all possible signings $\pm M$ of $M$. 
\end{definition}
%%%%
\begin{theorem}
\label{thm:karger-spielman}
Let $M$ be an $m\times|{\cal C}|$ nonnegative matrix as before, and let 
$\lambda$ be a nonnegative weighting of $\cal C$. Then, for any 
$1>\epsilon >0$, there exists (and is efficiently constructible) 
another nonnegative weighting $\lambda'$ of $\cal C$ such that $M\lambda'$ 
approximates $M\lambda$ up to a multiplicative factor 
of $1 \pm \epsilon$, and $|\supp(\lambda')| = O(\,\srk(M)\,/\,\epsilon^2)$.

If $M$ is Boolean, a different construction yield the same with 
$|\supp(\lambda')| \,=\, O(\,\trk(M) \cdot \log m\,/\,\epsilon^2)$.
\end{theorem}
Since $|{\cal C}|$ can be arbitrarily large or even infinite, 
"efficiently constructible" requires further explanation.
The input to the procedure is not the entire $M$ and $\lambda$, but only the
the nonzero values of $\lambda$, and the columns of $M$ corresponding
to them. The complexity is measured in terms of this
input. We further comment that $\supp(\lambda') \subseteq  \supp(\lambda)$. 

We are now ready to address the dimension reduction for $d$-volumes.
We start with general $d$.
\begin{theorem}\label{thm:main_reduction}
Let $\v$ be an $\ell_1$ $d$-volume on $n$ points, and 
let $0 < \epsilon < 1$ be a constant.
Then there exists (and is efficiently constructible) an $\ell_1$ $d$-volume 
$\v'$ that distorts $\v$ by at most a multiplicative factor of 
${1+\epsilon}\over {1-\epsilon}$, and the cut-dimension of $\v'$ is at most 
$O(n^d \log n/\epsilon^2)$, thus improving the trivial $O(n^{d+1})$.
\end{theorem}
\begin{proof} 
Let $M$ be a ${n \choose d+1} \times |{\mathcal C}|$ Boolean matrix whose rows 
are indexed by $d$-simplices, the columns are indexed by $d$-hypercuts, and 
$M(\sigma,C)=1$ if $\sigma$ belongs to the cut $C$ and $0$ otherwise. 
Observe that $M\lambda$'s correspond to $\ell_1$ $d$-volumes on $K_n^{(d)}$, 
and $|\supp(\lambda)|$ is an upper bound on the cut-dimension of the respective
$d$-volume. Thus, Theorem~\ref{thm:karger-spielman} applies, yielding an upper
bound of $O(\,\trk(M) \cdot d\log n\,/\,\epsilon^2)$ on the cut dimension.
It remains to upper-bound $\trk(M)$. It turns out be at most ${n-1} \choose {d}$.

Indeed, let $Q$ be a square $N\times N$ lower triangular nonsingular minor of $M$.  
Let the rows be indexed by $\{\sigma_i\}_{i=1}^N$, and the columns be
indexed by $\{C_i\}_{i=1}^N$ in this order. It means, in particular,
that $\sigma_i \in C_i$, but $\sigma_i \notin C_j$ for $j>i$.
We claim that the set of $d$-simplices $\{\sigma_i|~i=1, \ldots ,N \}$ does
not contain $d$-cycles. Indeed, assume by contradiction that it does
contain a cycle $Z$, and $r$ be the largest index such that $\sigma_r
\in Z$. Consider the corresponding $d$-cut $C_r$. Since $\sigma_r \in
Z\cap C_r$, by Claim~\ref{cl:cut-cyc}, $C_r$ must contain another
$d$-simplex from $Z$, contrary to the fact $\sigma_i \notin C_r$ for
every $i < r$. 

Thus, $\{\sigma_i|~ i=1, \ldots, N\}$ is acyclic, and $N$ is
bounded by the size of the maximum acyclic subcomplex, i.e., $d$-tree,
which is ${n-1} \choose {d}$.
\end{proof}

The special case of $d=1$ is precisely the much studied problem of
dimension reduction for $\ell_1$-metrics. While the elegant lower bounds 
of~\cite{brikman-charikar,lee-naor} show that one may at best hope for polynomial
(and not logarithmic) dimension reduction, the best known upper 
bound of~\cite{schechtman87} asserts that $c_\epsilon n\log n$ dimensions 
suffice for $1+\epsilon$ distortion. Theorem~\ref{thm:main_reduction} yields the same 
upper bound, however it strengthens~\cite{schechtman87} by claiming it for 
cut-dimension, which is larger than the usual geometric dimension
of the host $\ell_1$-space. Further improvement is provided by using a different method.
%%%%%%%
\begin{theorem}\label{thm:reduction-spielman}
Let $0 < \epsilon < 1$ and let $d$ be an $\ell_1$-metric on $n$
points. Then, there exists (and is explicitly constructible) an $\ell_1$- metric $d'$ 
such that $\dist(d,d') \leq {{1  + \epsilon} \over {1  - \epsilon}}$, while 
the cut-dimension of $\d'$ is at most $O(n/\epsilon^2)$.
\end{theorem}
\begin{proof} Let $M$ be the ${n \choose 2} \times \C$ Boolean matrix 
as in the proof of Theorem~\ref{thm:main_reduction} with $d=1$.
We claim that $\srk(M)$ is at most $n$.
This, in view of Theorem~\ref{thm:karger-spielman}, yields the desired bound. 

Let $B$ be an $|\C| \times n$ matrix whose rows are indexed by cuts,
and columns by vertices. For a cut $C=E(A,\overline{A})$, let  $B(C,v)=1$ if 
$v \in A$, and $-1$ otherwise. Let $X$ be a $n \times {n \choose 2}$ matrix with rows
indexed by $V$ and columns indexed by arbitrarily directed
edges. Let $X(v,e)=0.5$ if $v$ is the source of $e$,  $X(v,e)=-0.5$ if $v$ is the sink
of $e$, and $X(v,e)=0$ otherwise. Observe that $(BX)^T=\pm M$, and $\rank(BX) \leq n$. 
\end{proof}

Interestingly, $M$ has a full rank, as follows from Claim~\ref{cl:cutcone},
and thus $M$ is an example of a Boolean matrix with $\srk(M)$ roughly the square root 
of its rank. Note that by a standard 
%in: what is this ??
tensor product argument, $\srk(M)$ 
can never be smaller than that.

One may wonder how tight is the bound of Theorem~\ref{thm:reduction-spielman}. 
As we shall see, in terms of the dependence in $n$ it is best possible.
% even if one allows using cuts that are not necessarily in a 
%subfamily of the initial set.
\begin{theorem}
\label{th:lower}
Let $d_{n+1}$ be the shortest path metric of the unweighted path $P_{n+1}$, i.e., 
$d_{n+1}(i,j) = |i-j|$. This is certainly an $\ell_1$ metric. However, any 
metric $d' = \sum_{C\in \mathcal{C}'} \lambda_C \cdot \delta_C$ where~
$|\mathcal{C}'| \leq n/t$ ~distorts $\;d\;$ by at least ~$t$.
\end{theorem}
\begin{proof}
Since multiplicative distortion is not sensitive to scaling, we may assume
without loss of generality that $d$ dominates $d'$. This implies that
each $\lambda_C$ is at most 1, as $C$ must separate some pair of adjacent 
vertices $k-1,k$, and $1 \;=\; d(k-1,k) \;\geq\; d'(k-1,k) \geq \lambda_C$. 
But then all the distances in $d'$, and in particular $d'(1,n+1)$, are at most
$|\mathcal{C}'| = n/t$, and the statement follows.
\end{proof} 

Finally, our third dimension-reduction result is about geometrical $\ell_1$
$d$-volumes. Since for $d=1$ all hypercuts are geometrical, it is a nontrivial 
generalization of Theorem~\ref{thm:reduction-spielman}.
%%%%%%%%%%
\begin{theorem}\label{thm:reduction-spielman+}
Let $\v$ be a geometric  $\ell_1$ $d$-volume on $n$ points, and 
let $0 < \epsilon < 1$ be a constant.
Then there exists (and is efficiently constructible) a geometric $\ell_1$ 
$d$-volume $\v'$ that distorts $\v$ by at most a multiplicative factor of 
${1+\epsilon}\over {1-\epsilon}$, and the cut-dimension of $\v'$ is at most 
$O(n^d /\epsilon^2)$, thus improving Theorem~\ref{thm:main_reduction} in this
important special case.
\end{theorem}
\begin{proof} 
Consider the ${n \choose {d+1}} \times |{\mathcal C}|$ Boolean matrix as in the 
proof of Theorem~\ref{thm:main_reduction}, only this time ${\mathcal C}$ is the
family of all {\em geometrical} hypercuts. Call this matrix $P$. 
Since by Theorem~\ref{th: NEG} a geometrical $d$-hypercut volume is a real 
$d$-coboundary of $K_n^{(d)}$, we conclude that for every $C \in {\mathcal C}$ there 
exists $x_C \in \R^{n \choose d}$ such that the $C$-column of $P$ is equal to 
$~\pm x_C^T M_d$, where $M_d$ is as in the definition of the real $d$-coboundary.
Forming a matrix $X$ from vectors $\{x_C\}_{C \in {\mathcal C}}$, we conclude that 
$P = \pm X^T M_d$. Hence, $\srk(P) \;\leq\; \rank(M_d) \;=\; {{n-1} \choose d}$. 
Thus, by Theorem~\ref{thm:karger-spielman} we obtain an upper bound of 
$O(n^d/\epsilon^2)$ on the cut-dimension of the approximating geometrical 
$\ell_1$ $d$-volumes.
\end{proof}

%
%
%
%%%
%%%%%%%%%%%%%%%%%%%%%%%%%%%%%%%%
\subsection{Some Remarks and Applications}
\label{VTIssues}
\subsubsection{High Dimensional Sparsifiers and Approximating Forms.}~ 
One of the main results of~\cite{spielman-sparsifier} claims that every (nonnegatively) 
weighted graph $G$ has a $(1 \pm \epsilon)$-sparsifier $G'$ of size $O(n/\epsilon^2)$.
That is, for every $x\in \R^n$, the two forms 
$F_G(x)\,=\,\sum_{\{i,j\}\in E(G)}w_{ij} (x_i - x_j)^2$ and 
$F_{G'}(x)\,=\,\sum_{\{i,j\}\in E(G')}w_{ij}^{'} (x_i - x_j)^2$ differ by at most 
$(1 \pm \epsilon)$ multiplicative factor, where $E(G') \subseteq E(G)$ and 
$|E(G')| = O(n/\epsilon^2)$. The authors of~\cite{spielman-sparsifier} further
argue that such sparsifiers of the complete graph $K_n$ have many common properties
with (almost optimal) regular expanders of degree $\approx 1/\epsilon^2$, and in fact should 
be treated as such, despite the weights and the irregular degrees.

Using a convexity argument, one can re-define sparsifiers as above in terms of metrics spaces: 
$G'$ is a sparsifier of $G$ as above iff the two forms $F_G(d)\,=\,\sum_{\{i,j\}\in E(G)}w_{ij} d(i,j)$ 
and $F_{G'}(d)\,=\,\sum_{\{i,j\}\in E(G')}w_{ij}^{'} d(i,j)$ are $(1 \pm \epsilon)$-close
for every negative type distance $d$ on $V(G)$. This simple observation already has interesting
consequences. E.g., it implies that in order to $(1+\epsilon)$ approximate the average distance of 
a metric of negative type, it suffices to query $O(n/\epsilon^2)$ values (according to the suitable 
$w'$), and thus can be done in sublinear time. This somewhat surprising corollary was established 
earlier for Euclidean metrics (a special case of negative type metrics) by using a 
different argument in~\cite{BGS}, in turn improving upon an earlier result of P.\,Indyk. 

The general framework of Section~\ref{sec:GEN} together with 
original argument of~\cite{spielman-sparsifier} allow to extend the above results to 
higher dimensions. 

\begin{theorem}
\label{th:sublin}
For every weighted simplicial complex $K$ of dimension $d$ there exists
a sparsifier $K'$ such that the two $d$-forms 
$F_K(v^{(d)})\,=\,\sum_{\sigma \in K} w_K(\sigma) v^{(d)}(\sigma)$ and
$F_{K'}(v^{(d)})\,=\,\sum_{\sigma \in K'} w_{K'}(\sigma)  v^{(d)}(\sigma)$
differ by at most $(1 \pm \epsilon)$ multiplicative factor on any
function $v^{(d)}$ of negative type, and $|K'| = O(n^d/\epsilon^2)$. 
I.e., $K'$ has a constant average degree (measured with respect to
$(d-1)$-simplices).
\end{theorem}
\begin{proof}
A proof based on Theorem~\ref{thm:karger-spielman} is quite natural here, but we
prefer the original argument of~\cite{spielman-sparsifier} on which the latter theorem 
is based. Keeping in mind that the functions of nonnegative type are nonnegative combinations of 
(entrywise) squares of real $d$-coboundaries, it suffices to establish the statement for squares of 
real $d$-coboundaries.

Recall that a real $d$-coboundary $B_x \in \R^{n \choose {d+1}}$ is defined by a vector 
$x\in \R^{n \choose d}$ by $B_x^T = x^T M_d$, where $M_d$ is the real incidence matrix as in 
Section~\ref{sec:GEN}. Thus, $F_K(B_x^2) \,=\, x^T\cdot(M_d W_K M_d^T)\cdot x$, where
$W_K$ is a diagonal ${n \choose {d+1}} \times {n \choose {d+1}}$ matrix indexed by $d$-simplices,
where $W_K(\sigma,\sigma) = w_K(\sigma)$. Applying Theorem~\ref{th:spiel1} to the matrix
$M_d W_K M_d^T = (M_d \sqrt{W_K}) \cdot  (\sqrt{W_K} M_d^T)$ we conclude that there is another
weighting $w'$ such that $|\supp(w')| = O(\rank(M_d)/\epsilon^2)$, and $x^T\cdot(M_d W_K M_d^T)\cdot x$
and $x^T\cdot(M_d W' M_d^T)\cdot x$ differ by at most $(1 \pm \epsilon)$ multiplicative factor.
Keeping in mind that $\rank(M_d) \;=\; {{n-1} \choose d}$, and defining $K'$ as the support of $w'$, 
we arrive at the desired conclusion.
\end{proof}

As a bonus we get a sublinear algorithm for approximating the
average value of functions of negative type, in particular the Euclidean $d$-volumes,
and the geometric $\ell_1$ $d$-volumes:
\begin{corollary}
\label{cor:sublin}
In order to $(1+\epsilon)$ approximate the average value of a function of negative type,
it suffices to query $O(n^d/\epsilon^2)$ predefined (and efficiently computable) 
$(d+1)$-tuples forming a high-dimensional sparsifier of an (constant) average 
degree $\approx~1/\epsilon^2$.
\end{corollary} 
%%%%%%%%%%%%%%%%%%%%%%%%%%%%%%%%%%%%%%%%%%
%%%%%%%%%%%%%%%%%%%%%%%%%%%%%%%%%%%%%%%%%%%%
\subsubsection{Sparse Spanners.}~ It is well known that the average degree in a graph $H$ with 
$n$ vertices and girth $g$ is $n^{O\left( {1\over g} \right)}$. Since (see~\cite{dobkin}) 
the shortest-path metric $d_G$ of a weighted graph $G$ can be $(g-1)$-approximated 
by that of its subgraph $H$ of girth $g$, there exists a $g$-spanner of $G$ with
at most $n^{1 + O\left({1\over g}\right)}$ edges. The construction naturally carries on to volumes,
which brings us to a question: What is the maximal number of $d$-simplices in a simplicial 
complex $K$ on  $n$ vertices, such that the smallest $d$-cycle of $K$ is of size $\geq g$? 
The probabilistic construction of Lemma~\ref{lm:expander} (with small local amendments) 
shows that for $d=2$ there exists $K$ of average degree $O(\log n)$, and the smallest 
cycle of size $\tilde{\Omega}(n^{0.2})$. (By degree of a $1$-simplex $e$ we mean the number of 
$2$-simplices in $K$ that contain $e$.) 
Thus, the situation for $d=2$ significantly differs
from the graph theoretic case. 
It would be interesting to get tighter bounds for this problem.
See also~\cite{LuMe} for a somewhat related discussion.
%%%%%%%%%%%%%%%%%%%%%%%
\subsubsection{On ${\mbox{\bf c}}_{\mbox{\bf\small 1}}$(K).}~ Like in graphs, given a $d$-complex $K$ one may ask
what is the worst possible distortion of approximating $\v_K$, a lightest-cap volume
of $K$ (over all choices of nonnegative weights of its simplices), by an $\ell_1$ 
volume. This important numerical parameter is called (by analogy with graphs) $c_1(K)$. 
One of the most important open questions in the theory of finite metric spaces
is whether any graph $G$ lacking a fixed minor has a constant $c_1(G)$ 
(see e.g.,~\cite{GNRS} for a related discussion and partial results). 
It is natural to ask a similar question about $d$-complexes: what properties
of $K$ would imply a nontrivial upper bound on $c_1(K)$? The techniques 
of~\cite{GNRS} imply this: $c_1 (K) \leq 2^{\chi(K)}$, where $K$ 
(as usual) is assumed to have a complete $(d-1)$ skeleton and
$\chi(K)$ is the Euler characteristic of $K$.
The construction proceeds via repeatedly picking a minimal cycle, and removing 
a random $d$-simplex in it with probability proportional to its volume. The lightest-cap
volume of the random (sub-)hypertree of $K$ obtained in this manner 
dominates $\v_K$, yet stretches it (in expectation) by only a constant factor.
%%%%%%%%%%%%%%%%%%%%%%%%%%%%%%%%%%%%%%%%%%%%%%%%%%%%%%%%%%%%%%%%%%%%%%%%%%%%%%%%%%%%
%
%%%%%%%%%%%%%%%%%%%%%%%%%%%%%%%%%%%%%%%%%%%%%%%%%%%%%%%%%%%%%%%%%%%%%%%%%%%%%%%%%%%%%
\section{Abstract Sparsification Techniques}
As already indicated above, the general problem to be discussed is in this part of the paper 
is as follows.  Let $\cal C$ be a family of nonnegative vectors in  $\R^m$, and let $\cone({\cal C})$ 
be the convex cone spanned by it. The goal is, given a vector $w \in \cone({\cal C})$, to produce a 
small subset ${\cal C}' \subset {\cal C}$ and a vector $w' \in \cone({\cal C}')$ that (pointwise) 
approximates $w$ up to a multiplicative factor of $1 \pm \epsilon$. 

Using the matrix notation, let $M$ be a $m\times |{\cal C}|$ real nonnegative matrix. 
Then, given a nonnegative vector $\lambda \in \R^{|{\cal C}|}$, the goal is to produce a new 
$\lambda' \in \R^{|{\cal C}|}$ such that on one hand $w'=M\lambda'$ approximates  
$w=M\lambda$ up to a multiplicative factor of $1 \pm \epsilon$, and on the other hand
$\lambda'$ has small support. The columns of $M$ are the vectors of 
$\cal C$, and $\lambda,\lambda'$ are coefficients of nonnegative combinations 
of these vectors. For computational purposes, we assume that the input to the procedure 
is not $M$ and $\lambda$, but only the the nonzero values of $\lambda$, and the columns of $M$ 
corresponding to them. It will always hold that $\supp(\lambda') \subseteq  \supp(\lambda)$. 

We seek to single out the relevant parameters of the matrix $M$ such that $|\supp(\lambda')|$ 
as above can be upper-bounded in terms of these parameters alone, not depending on  $\lambda$. 
The problem appears to be of a fundamental nature, far transcending the particular context of
the previous sections (there are some additional examples at the end of this section). 
We initiate the study of this problem here, and produce two families of such parameters yielding the 
desired upper bounds. The first result is restricted to Boolean
matrices, the other is more general but  weaker (if one ignores a $\log m$ factor,
which in fact is not always ignorable).  Both results are almost tight in the special case, 
sufficient, but apparently not necessary. Importantly, they are inherently limited to $0 < \epsilon < 1$.
The situation for  large $\epsilon$'s appears to be radically different, and calls for further study.

%After stating and proving the two results, we shall discuss some examples, and also establish some
%relations with geometric discrepancy. 

In what follows, it will be convenient and combinatorially justified
to interpret $M$ as $M_{|{\cal F}| \times |{\cal C}|}$, an
'incidence' matrix of a quantitative relation between the members of a
family $\cal F$ (indexing the rows) and the family $\cal C$ (indexing
the columns). In this interpretation, $\lambda=\{\lambda_c \}_{c \in
  {\cal C}}$ is a weighting of $\cal C$ that induces a weighting
$w=\{w_f \}_{f \in {\cal F}}$ on $\cal F$ by assigning $w(f) =
\sum_{c\in {\cal C}} M(f,c) \lambda_c$. I.e., $w(f)$ is the weighted
sum of all the columns incident to $f$. For example, in
Theorem~\ref{thm:main_reduction}, ${\cal F}$ stands for the family of
$d$-simplices, and $\cal C$ stands for the family of $d$-hypercuts.
The relation represented by the corresponding $M$ is the membership:
$M(\sigma,C) = 0$ if $\sigma \in C$, and $M(\sigma,C) = 0$ otherwise.
%%%%%%%%%
%
\subsection{The First Technique}
We restrict our attention to Boolean matrices $M$.
The key parameter of $M$ will be its {\em triangular rank}. Recall that
the triangular rank of $M$, $\,\trk(M)$, is the size of 
the largest nonsingular lower-triangular square minor of $M$. 
The rows and the columns of the minor may appear in order different from 
that of $M$.
\begin{theorem}
\label{thm:karger-new}
Let $M$ be a 0/1 matrix as before, $\lambda$ a nonnegative weighting of $\cal C$, and 
$w=M\lambda$. Then, for any $0<\epsilon<0$, there exists (and is efficiently constructible) 
another nonnegative weighting  $\alpha$ of $\cal C$, such that the support 
of $\alpha$  is of size at most $O(\,\trk(M) \cdot \log m|\,/\,\epsilon^2)$, and
$w'=M\alpha$ (entrywise) distorts $w$ by at most $(1 \pm \epsilon)$ multiplicative factor.
\end{theorem}
%%%%%%%%%%%%%%%%%%%%%
%%%%%%%%%%%%%%%%%%%
\begin{proof}%{\bf (Of Theorem \ref{thm:karger-new}).} 
The method of proof is inspired by the method of Karger and Bencz{\'u}r 
from~\cite{karger-benczur}. 

The existence of $\alpha$ will be established using a probabilistic argument. 
We start with some preparatory observations and tools. Let $\Col_c$
be the column of $M$ indexed by $c \in {\cal C}$.
Making $\lambda_c$ copies of each column $\Col_c$, $c \in {\cal C}$,
we arrive at the new $M'$ with same triangular rank, and $w=M \overline{1}$,
i.e., $\lambda$ becomes an all-1 vector, and $w$ is the sum of columns. 
We assume that w.l.o.g., this  is the original input. 
(Of course, $\lambda_c$ may not be integer, but we take for the sake of the proof infinitesimal units, 
and use the scalability of the problem. The algorithmic issues will be addressed later.). 
In addition, w.l.o.g., we assume that $M$ does not have all-0 columns.

As we are about to sample the columns of $M$, notice that some columns are 
more essential for $w$ than the others, and thus the sampling is necessarily 
non-uniform. For example, if a certain column $\Col_c$ is the only column of 
$M$ such that $\Col_c(f) > 0$ for some $f \in {\cal F}$, and $w_c > 0$,  
then $\Col_c$ must necessarily be chosen. More generally, if the row of some
$f \in {\cal F}$ has small support, the columns corresponding to this support
should be sampled with relatively hight probability. This motivates the following definition, 
analogous to the {\em strength of an edge} in~\cite{karger-benczur}:

\begin{definition}[The strength of a column]
\label{def:strength}
Let $M$, $w=M\overline{1}$, be as above. The function 
$s: \C \mapsto \N$ assigning to each column of $M$ a {\em strength} value,
is defined by the following iterative process:
\\ 
Let  $M_\ast = M$; $w_\ast = w$,  and  $m = \min^{+}_{f \in \F} w(f)$,
where $\min^+$ is the smallest strictly positive entry of $w$.

While $M$ is not all 0, repeat: 
\begin{enumerate}
%%%%%%%%
\item  While there is $f \in \F$ such that $0 < w_\ast(f)\leq m$, do: \\
Assign $s(c)=m$ for every $c$ in the support of the $f$-row,
$\Row_f$ of $M_\ast$.  For every such $c$ set
 $\Col_c$ to $0$ to get a new $M_\ast$, and update 
$w_\ast$ to the new sum of columns of $M_\ast$.
%%%%%%%
\item If $w_\ast$ is not identically 0, set $m = \min^{+}_{f \in \F} w_\ast(f)$, 
and return to (1).
\end{enumerate}
\end{definition}
Observe that while the order in which $f$'s are chosen in (1) is 
somewhat arbitrary, at each invocation of (2) the set of columns 
set to 0, and the new value of $w_\ast$ are uniquely defined, and do not 
depend on the order of choices made in (1). Thus the strength
function is well defined. Observe also that identical columns
necessarily get identical strengths. Finally, observe that the value of
the strength never decreases along the run of the process above.
\begin{definition}\label{def:strength2}
Let $\mathcal{C}$ be the column indices as above, and let 
$s_1 < s_2 < \ldots < s_t$ be the sequence of corresponding strengths
in the increasing order. Define 
$\mathcal{C}_i = \{c \in \mathcal{C}\;|~ s(c) \geq s_i \}$, and $w_i = \sum_{c\in \mathcal{C}_i} \Col_c$,  ~$i=1,2,\ldots,t$. Observe
that $\mathcal{C}_i$ is monotone decreasing with respect to containment,
and that all the non-zero entries of $w_i$ are at least $s_i$.
\end{definition}
Call a single run of the while loop of (1) a {\em phase}. During a phase, 
one sets to $0$ precisely all the (still surviving) columns $\Col_c$ 
such that $c \in \supp(\Row_f)$ in $M_\ast$, causing $w_\ast(f)$ to become $0$. 
All these columns get the same strength $m$. 
The following Lemma establishes some important properties of the strengths.

\begin{lemma}
\label{lem:11}
$\mbox{}$ \\
 \begin{enumerate}
\vspace*{-2.5mm}\item  Let $s_k$ be the maximal strength of a column
         $\Col_c$ where $c \in \supp(f) \subseteq {\mathcal  C}$ in the original $M$.
         Then,  $|\supp(\Row_f)| \geq s_k$. In particular, this implies that 
         $w(f) \geq s_k$. 
         Observe, however, that by maximality of $s_k$, during any single 
         phase no more than $s_k$ $c$'s from $\supp(\Row_f)$ are set to $0$. 
 \vspace*{-2mm}\item  $\sum_{c \in \mathcal{C}} \frac{1}{s(c)} \;\leq\; N$, 
        where $N$ is the total number of phases. 
        This parameter is crucial for the forthcoming analysis.
 \vspace*{-2mm}\item  The total number of phases $N$ is at most $\trk(M)$.
\end{enumerate}
\end{lemma}
\begin{proof}
The first statement directly follows from the definition of the
strengths. That is, let $c \in \supp(\Row_f)$ for which
$s(c)=s_k$. Then when $s(c)$ is set, $w_\ast(f) = s_k \leq
w(f)$. Hence, since each $c' \in \supp(\Row_f)$ contributes exactly
$1$ to $w_{\ast}(f)$ the claim follows.
%Since $\supp(\Row_f)$ in $M$ contains some $c$' cuts of strength
%$s_k$, and does not contains $c'$ strength, $\supp(\Row_f)$ 
%contains precisely $d_k(v,u)$ column indices cuts of strength $s_k$. 
%Recall that $w(f) \geq w_k(f) \geq s_k$. 

For the second statement, consider a contribution of a phase of (1) to the left 
hand side of the inequality. Each column set to 0 contributes ${1\over s_i}$, 
where $s_i$ is the current $m$ (constant during the phase), while the number of 
such columns is $w_\ast(f)$ at the beginning of the phase, which is at most $m$. 
Thus, the contribution of a phase is at most $s_i \cdot {1\over
  {s_i}}\;\leq\;1$, which implies the claim.
%%%%%%%%%%%%%%
\ignore{
For the third statement, observe that every phase terminates in a contraction of an edge,
i.e., in the beginning of the phase $d_\ast(x,y) > 0$, and after its termination $d_\ast(x,y)=0$.
Since after $n-1$ such contractions the space collapses to a single point, there are at
most $n-1$ phases. 
}
%%%%%%%%%%%%%%

For the third statement, for each phase $i$, let $f_i \in \F$ be the coordinate 
that initiated the phase. Mark a $c_i \in \C$ such that $\Col_{c_i}$ was set to $0$
during the phase. Consider the corresponding $N\times N$ minor of $M$. Clearly,
$M(f_i,c_i)=1$ for all $i$. Since during the $i$'th phase
all the surviving columns $\Col_c$ such that $M(f_i,c)=1$ are removed, 
it follows that for every $k > i$ and $c_k$ that survives after the $i$'th
phase, $M(f_i,c_k)=0$. Thus, the $N$-minor of $M$ on rows $(f_1,f_2,\ldots,f_N)$
and columns $(c_1,c_2,\ldots,c_N)$ is a nonsingular
lower triangular matrix.
\end{proof}
%%%%

We presently define the sampling procedure to be used in the proof of 
Theorem~\ref{thm:karger-new}:
\begin{definition}
\label{def:2}
Let $\rho > 1$ be a parameter to be defined later. 
For each $c \in \mathcal{C}$, define $p_c \;=\; \min\{\frac{\rho}{s(C)},1\}$,
and let $X_c$ be a random variable (indicating whether the column $c$ is chosen) defined by 
$\P(X_c=1) = p_c$ and \mbox{$\P(X_c=0)=1-p_c$}. Choosing the columns randomly and independently 
according to the specified probabilities, we obtain a random subset of
columns $\mathcal{C}'=\{c| ~X(c)=1 \}$. Finally, setting $\alpha_c =  1/p_c$, 
we define a random  vector $w' = \sum_{c \in \mathcal{C}'} \alpha_c \Col_c$.
\end{definition} 

The shall use here the following version of the Chernoff Bound 
(see Theorems A.1.12, A.1.13 in~\cite{alon-spencer}).
%%%%%%%%%%%%%%%%%
\vspace{-0.3cm} \begin{theorem}\label{thm:alon-spencer}\cite{alon-spencer}
Let $X_1, \ldots X_n$ be independent Poisson trials such that $\P(X_i=1)=p_i$. 
Let $S = \sum X_i$ and $\nu = \sum p_i$. Then  for any $0 <\beta < 1$,~
$\P[S \not\in (1 \pm \beta) \cdot \nu] \;\leq\;   2e^{-\beta^2 \nu/3}$.
\end{theorem}
%%%%%%%%%%%%%%%%%
We start with showing that almost surely 
the size of $\mathcal{C}'$ is $O(\rho N)$. 
\begin{lemma}\label{lem:2}
With probability $1-o(1)$ the size of $\mathcal{C}'$ is at most
$2\rho \cdot N \;<\; 2\rho n$.
\end{lemma}
\begin{proof}
Since $|\mathcal{C}'| = \sum_{C \in \mathcal{C}} X_C$, 
items (2) and (3) of 
Lemma~\ref{lem:11} imply that 
\[
E[|\mathcal{C}']|  = \sum_{c \in
\mathcal{C}} p_c ~\leq~ \sum_{c \in \mathcal{C}} \rho/s(C) ~<~ \rho \cdot N\,.
\] 
Since the $X_c$ are independent, Theorem~\ref{thm:alon-spencer}
applies, implying that $\P(|\mathcal{C}'| > 2\rho N)\;\leq\; 2e^{-2/3 \, \rho N}$.
\end{proof}

Next, observe that the expectation of $w'$ is $w$:
\begin{claim}\label{cl:expect}
$\E(w') = w$.
\end{claim}
\begin{proof}
$\E(w') = \E[ \sum_{c \in \mathcal{C}'} \alpha_c \cdot \Col_c] 
= \E[ \sum_{c \in \mathcal{C}} X_c \cdot  \alpha_c \cdot \Col_c] = 
\sum_{c \in \mathcal{C}} (p_c \cdot \alpha_c) \Col_c = 
\sum_{c \in \mathcal{C}} \Col_C = w\,.$
\end{proof}

The next goal is to show that $w'$ is tightly concentrated around its 
mean. Since the parameters $p_c$ and $\alpha_c$ of the column $c$ depend 
solely on its strength $s(c)$, the sequence of strengths $s_1 < s_2 < \ldots < s_t$
defines the sequence of probabilities  $p_1 \geq p_2 \geq \ldots \geq p_t$, and the sequence 
of weights $\alpha_1 \leq \alpha_2 \leq  \ldots \leq  \alpha_t$. The following
claim is easily verified; essentially it is an Abel's summation transform:
\begin{claim}
\label{cl:d'}

\[
w' ~=~ \sum_{\mathcal{C}'} \alpha_c \cdot \Col_c ~=~ 
\sum_{i=1}^t  \Delta_i \cdot \sum_{c \in \mathcal{c}_i} X(c) \cdot \Col_c\;,~~ 
\mbox{ where } ~~\Delta_i  = \alpha_i -  \alpha_{i-1}, ~~\mbox{ and } ~~\Delta_1 = \alpha_1.
\]
\end{claim}
 Let $z_i = \sum_{c \in \mathcal{C}_i} X_c \cdot \Col_c$.  
The key point in the forthcoming lemma is that the random component of $z_i(f)$ is either empty, 
or has expectation $\geq \rho$, making the Chernoff bound of Theorem~\ref{thm:alon-spencer}
applicable. Choosing $\rho$ appropriately, and using the union bound over all $i,f$, one arrives
at the desired conclusion.  
%%%%%%%%%%%%%%%
\begin{lemma}
\label{lem:conc}
Set $\rho = \frac{3}{\epsilon^2}\left( \ln (2|{\cal F}|) + \ln t +
  k \right)$, where $k>0$ is any real number, and $t$ is the number
of distinct strengths $s_i$.  
Then, $\P[\,w' \not\in (1 \pm \epsilon) w\,] \;\leq\;   e^{-k}$.
\end{lemma}
\begin{proof}
For any fixed $f \in \F$, $z_i(f) = \sum_{c \in \mathcal{C}_i} X_c \Col_c(f)$, 
a sum of independent Boolean variables. The columns with $s(c) \leq \rho$ 
deterministically contribute 1 to this sum, as in this case $p_c = 1$. 
If there are no other columns, we are done. Else, let $s_k>\rho$ be 
the maximal strength of the column in $\supp(\Row_i(f))$. By Lemma~\ref{lem:11}(1)
there must be at least $s_k$ columns of such strength in this collection, and therefore
$\E[z_i(f)) \;\geq\; s_k p_k \;\geq\; s_k \cdot {\rho \over s_k} \;=\; \rho$. 
Thus, by Theorem~\ref{thm:alon-spencer},
\begin{equation}\label{eq:toremark}
\P[\,z_i(f) \not\in (1\pm \epsilon) \cdot \E[z_i(f)]\,] ~\leq~ 
2e^{{{-\epsilon^2}\over 3}\cdot \E(z_i(f))} ~\leq~
2e^{{{-\epsilon^2}\over 3}\cdot \rho}\,.
\end{equation}
Substituting the proposed value for $\rho$, we conclude that the above
probability is at most $|\F|^{-1} \cdot N^{-1} \cdot e^{-k}$.
Taking the union bound over all $i=1,2,\ldots,t$ and $f\in\F$, we conclude that the probability that there 
exist $i,f$ with with $\P[\,z_i(f) \not\in (1\pm \epsilon) \cdot \E(\mu_i(x,y))\,]$
is at most $e^{-k}$. Keeping in mind that $w' = \sum_{i=1}^t  \Delta_i \cdot z_i$, the statement follows. 
\end{proof} 

Choosing $k$ large enough constant, and keeping in mind that $t \leq N
\leq \trk(M)$, Lemma~\ref{lem:2} implies that  $\mathcal{C}'$
is almost surely of size at most $2\rho N =
O(\trk(M)\log(|\F|)/\epsilon^2$.  On the other hand, by
Lemma~\ref{lem:conc}, $w' = \sum_{c \in \mathcal{C}'} \alpha_c \cdot
\Col_c$ almost surely distorts $w$ by at most a
$(1+\epsilon)/(1-\epsilon)$ multiplicative factor.  This establishes
Theorem~\ref{thm:karger-new}.
\end{proof}

{\bf Algorithmic considerations:}
Recall that for simplicity of presentation, instead of working with a weighted set
$\C$, we have worked with unit-weighted multiset obtained by producing
$\lambda_C$ duplicates of each $C$.  Due to scalability, we could
assume that $\lambda_C$ is a huge integer, and the rounding issue does
not arise. While it indeed simplifies the presentation, this approach results
in a very inefficient randomized procedure for selecting the sparser
family $\C'$ of Theorem~\ref{thm:karger-new}. However, observe that the
duplicates of $C$ are sampled randomly and independently with the same
probability $p_C$, the resulting total weight of $C$ is distributed
according to a binomial distribution, and can be efficiently produced.
When weights are not integers, we may simulate the process by massive
scaling, which leads to sampling according to the Poison distribution with
parameter $\lambda_C$. A detailed discussion of this issue can be
found in~\cite{karger-benczur}, (see Section 2.4 and Theorem A.1
there). The resulting sparsification procedure can be implemented in
time $O(n^2 \cdot |\C|)$. 

To conclude the discussion of this section, let us remark that for large $\F$'s,
sometimes a better upper bound can be obtained, as in the original result of~\cite{karger-benczur},
by strengthening the Eq.~(\ref{eq:toremark}) in the proof of Lemma~\ref{lem:conc}. Instead
of using a uniform lower bound on the expected value of $z_i(f)$, one may sometimes rely on finer 
distributional properties of this random variable, and get significantly stronger results.
%%%%%%%%%%%%%%%%%%
%
\subsubsection{The Second Technique}
Here $M$ does not have to be Boolean, just nonnegative. The key parameter of $M$ will 
be, as in Theorem~\ref{thm:karger-spielman}, the minimum possible rank
of (Hadamard) square root of $M$.
%%%
\begin{definition}\label{def:rigidity}
For $D \geq 1$,  define $\srk_D(M)$ as the minimum rank over all 
matrices $A$ such that for all $i,j$, it holds that $M_{ij}\,\leq\, A_{ij}^2 \,\leq\,D \cdot M_{ij}$.
Equivalently, $M \,\leq\, Y\circ Y \,\leq \,D \cdot M$, where $\circ$ stands for the Hadamard 
(i.e., entrywise) product of matrices. In particular, let $\srk(A)=\srk_1(A)$. 
\end{definition}
\begin{theorem}\label{thm:general_tech2}
Let $M$ be a matrix as before, $\lambda$ a nonnegative weighting of $\cal C$, and 
$w=M\lambda$. Then, for any $0<\epsilon<0$, there exists (and is efficiently constructible) 
another nonnegative weighting  $\lambda'$ of $\cal C$, such that the support 
of $\lambda'$  is of size at most $O(\,\srk_D(M) \,/\,\epsilon^2)$, and
$w'=M\alpha$ (entrywise) satisfies 
$(1-\epsilon)\cdot M\lambda \;\leq\; M\alpha  \;\leq\;  D\cdot(1+\epsilon)\cdot M\lambda\,$.
\end{theorem}
Observe that $~{\srk_D(M)}\;\geq\; \trk(M)$ for any $D$.

The powerful technical tool we are going to employ, (implicitly) 
appears in its strongest form in recent important paper~\cite{spielman-sparsifier}:
%%%%%%%%%%%%%%
\begin{theorem}\cite{spielman-sparsifier}
\label{th:spiel1}
Let $B_{m\times n}$ be a real valued matrix, and let $Q_{n\times n}$ be $Q=B^T B$.
Then, for every $\epsilon > 0$ there exists (and can be efficiently constructed)
a nonnegative diagonal matrix $A_{m \times m}$ with at most $O(\epsilon^{-2} n)$ 
(or even $O(\epsilon^{-2} \rank(Q)\,)$ positive entries, and with following property.
Let $\tilde{Q} = B^T A B$. Then, for {\em every} $x\in\R^n$
it holds:
\[
(1-\epsilon)\cdot x^T \tilde{Q} x ~\leq~ x^T Q x   ~\leq~   (1+\epsilon)\cdot x^T \tilde{Q} x~. 
\]
\end{theorem}
Actually,~\cite{spielman-sparsifier} is solely interested in the Laplacian matrices of positively weighted
graphs, and the above theorem is stated there only for such $Q$'s. However, a close 
examination of the proof reveals that with a minor change (related to rank of $Q$) it works also 
for general positive semidefinite symmetric $Q$'s. 
\begin{proof}
Clearly, it suffices to the prove the theorem for $D=1$.
The extension for larger $D$'s is obtained in a trivial manner.

Let $k=\srk(M)$. Our aim is to attach to each $f \in \F$ a vector $x_f \in \R^k$, and to each 
$c \in \C$ a vector $b_c \in \R^k$ such that $x_f \cdot b_c = \pm M(f,c)^{1\over 2}$. 
Let $B(\lambda)$ be a $|\C| \times k$ matrix whose rows are $\sqrt{\lambda_c} b_c$. Then, for
each $f$, it holds that 
\[
x_f^T B(\lambda)^T B(\lambda) x_f ~=~ \sum_{c\in \C} \lambda_c \left(  M(f,c)^{1\over 2}
  \right)^2 ~=~ \sum_{c\in \C} \lambda_c M(f,c) ~=~ w(f)\,.
\]
Applying  Theorem~\ref{th:spiel1} to the matrix
$B(\lambda)^T B(\lambda)$, we get the desired $k/{\epsilon^2}$ sparsification.

It remains to construct the required $x_f$'s and $b_c$'s. Equivalently, we need to
construct the matrices $B_{|{\cal C}|\times k}$ and $X_{k\times |{\cal F}|}$ such that 
$(BX)^T \circ (BX)^T ~=~ M$. However, it is given that there exists (and can be 
efficiently found) a matrix $Y_{|{\cal F}| \times |{\cal C}|}$ of rank $k$ such that 
$Y\circ Y = M$.  Using a standard  linear algebra argument, we (efficiently) decompose $Y^T$ 
as $Y^T=BX$, where $B$ and $X$ are matrices as required above.    
\end{proof}
%%%%%%%%%%%%%%%%%%%%%%%%%%%%%%
%%%%%%%%%%%%%%%%%%%%%%%%%%%%%
\subsection{Additional Remarks and Examples}
\label{sec:last}
%&&&&&&&&&&&&
\ignore{
\label{thm:karger-new} (techniques)
\begin{theorem}\label{thm:main_reduction}
Let $\v$ be an $\ell_1$ $d$-volume on $n$ points, and 
let $0 < \epsilon < 1$ be a constant.
Then there exists (and is efficiently constructible) an $\ell_1$ $d$-volume 
$\v'$ that distorts $\v$ by at most a multiplicative factor of 
${1+\epsilon}\over {1-\epsilon}$, and the cut-dimension of $\v'$ is at most 
$O(n^d \log n)/\epsilon^2)$, thus improving the trivial $O(n^{d+1})$.
\end{theorem}
\begin{proof} 
Let $M$ be a ${n \choose d+1} \times |{\mathcal C}|$ Boolean matrix whose rows 
are indexed by $d$-simplices, the columns are indexed by $d$-hypercuts, and 
$M(\sigma,C)=1$ if $\sigma$ belongs to the cut $C$ and $0$ otherwise. 
Observe that $M\lambda$'s correspond to $\ell_1$ $d$-volumes on $K_n^{(d)}$, 
and $|\supp(\lambda)|$ is an upper bound on the cut-dimension of the respective
$d$-volume. Thus, Theorem~\ref{thm:karger-spielman} applies, yielding an upper
bound of $O(\,\trk(M) \cdot d\log n\,/\,\epsilon^2)$ on the cut dimension.
It remains to upper-bound $\trk(M)$. It turns out be at most ${n-1} \choose {d}$.
%&&&&&
}
%%%%%%%%%%%%%%%%%%%%%%%%%%%%%%%
\subsubsection{\bf co-Circuits in  Matroids.}
The argument used for bounding the triangular rank of $M$ employed in the proof of
Theorem~\ref{thm:main_reduction} actually applies in the much more general case 
when the rows of $M$ are indexed by the elements of a matroid $\cal M$, 
and the columns of $M$ are indexed by co-circuits (or circuits) of $\cal M$'s. 
One needs only to observe that the intersection of a circuit and a co-circuit in
$\cal M$ cannot be a single element. The conclusion is that in the general case, 
$\trk(M)$ is at most the size of a maximum independent set in ${\cal M}$.   
%\\ \\
%%%%%%%%%%%%%%%%%%%%%%%%%%%%%%%
\subsubsection{\bf Splitting Set Systems.} The Boolean matrix $M$ used in the proof of Theorem~\ref{thm:karger-new} 
for $d=1$ (i.e., the inclusion matrix of edges vs.~edge cuts) could be described somewhat
differently using vertices instead of edges. Then, the rows correspond to subsets $e$ of $V$ of size 2,
the columns correspond to nontrivial subsets $A$ of $V$, and $M(e,A) = 1$ iff $|e \cap A|=1$.
This situation is a special case of what we call a {\em splitting set system}, and the claim
that $\trk(M) = |V|-1$ turns out to be a special case of a more general theorem.  

Let $\F, \C \subseteq 2^V$ be any two families of subsets of $V$. For
every $f \in \F$ and $c \in \C$, say that $c$ splits $f$ if $c\cap f \,\neq\,\emptyset$ and
$\bar{c} \cap f\,\neq\, \emptyset$. Define the incidence matrix
$M$ by $M(f,c) = 1$ if $c$ splits $f$, and $M(f,c) = 0$ otherwise. 
%Note,
%for $\F$ being the set of pairs and $\C$ being the set of cuts, this
%is just the matrix $M$ from the proof of Theorem~ \ref{thm:karger-new} above.
\begin{theorem}
\label{th:split}
Let $M$ be the incidence matrix as above. Then, $\trk(M) \leq |V|-1$.
\end{theorem}
\begin{proof}
Let $Q$ be a square $N\times N$ lower triangular nonsingular minor of $M$.  
Let the rows be indexed by $\{f_i\}_{i=1}^N$, and the columns be
indexed by $\{c_i\}_{i=1}^N$ in this order. It means, in particular,
that $c_i$ always splits $f_i$, but $c_j$ with $j>i$, does not split
$f_i$. Consider the partition of $V$, the underlying set induced by 
the family $\{c_{i+1},\ldots, c_N \}$. Since no $c_j$ in it splits $f_i$, 
$f_i$ must be contained in a single atom of the partition. Since $c_i$
splits $f_i$, the partition induced by $\{c_i, c_{i+1},\ldots, c_N \}$
must strictly refine the previous partition. Therefore, the number of
atoms in the partition induced by $\{c_1, c_2,\ldots, c_N \}$ is at
least $N+1$. But then $N+1 \leq |V|$, and the statement follows.
\end{proof}
%%%%%%%
\subsubsection{\bf Random Boolean Matrices.} Let $M$ be a random $m\times n$ Boolean matrix, $m\geq n$
Then, by a standard probabilistic method argument, $\trk(M) = \theta(\min\{\log(m),n\})$ 
almost surely. The trivial details are omitted.
%\\ \\
\subsubsection{\bf An Application to Geometric Discrepancy.}
We conclude the paper with an example of an application of the sparsification 
methods of this section to a natural purely geometric question with a discrepancy 
flavor.

The general problem is as follows. Assume we have a family $\cal F$ of bodies
in $\R^d$. The goal is to produce a small {\em sampling set} $P \subset \R^d$, i.e., 
a set of points with associated positive weights, such that for every body $B \in {\cal F}$
it holds that $\sum_{p\in P\cap B} w_p = (1 \pm \epsilon) \v^{(d)}(B)$, where $\v^{(d)}$
is the Euclidean volume. Unlike the usual discrepancy setting, bodies of
small volume are as important as bodies of large volume.
%%%%%%%%%%%%%%%%%%%%%%%%%
\begin{theorem}\label{thm:discrepancy}
Let $S$ be a set of $n$ points in the plane, and let $\cal F$ be
family of all closed non self-intersecting polygons with vertices
in $S$. Then, there exists a sampling set $P$ for $\cal F$
as above of size $O(n^2\,/\,\epsilon^2)$. Moreover, such $P$ can efficiently 
constructed in time polynomial in $n$.
\end{theorem}
%%%%%%%%%%%%%%%%%%%%%%
\begin{proof}%{\bf (Sketch)}
First, observer that it suffices to establish the theorem for the triangles with 
vertices in $S$, since all other polygons in $\cal F$ can be triangulated, and thus are
disjoint union of such triangles (ignoring the boundaries).  Treating
these triangles as 
a 2-dimensional realization of $K_n^{(2)}$, and associating with each point $p\in \R^2$
a geometrical 2-hypercut (as in the proof of Theorem~\ref{cl:Eu-gcuts}) we conclude that
the induced Euclidean volume on  $K_n^{(2)}$ is a geometrical $\ell_1$ volume. Thus,
by Theorem~\ref{thm:reduction-spielman+},  this 2-volume can be $(1\pm\epsilon)$ approximated
by a geometrical $\ell_1$ 2-volume of cut-dimension $O(n^2\,/\,\epsilon^2)$.  Moreover, since 
$\;\supp(\lambda') \subseteq  \supp(\lambda)$, the approximating $\ell_1$ 2-volume is
induced by a weighted sampling set of points $P$ of this size.

In order to produce $P$ in polynomial time, first compute the  $O(n^4)$ cells created the lines 
spanned by $S$. The initial sampling set $P_0$ will have a point $p$ in the interior of each such 
cell, with the associated weight $w_p$ being the area of the
cell. Clearly,  samples $P_0$ without errors, 
but it is too big. Next, apply the procedure underlying Theorem~\ref{thm:reduction-spielman+} to 
this input to obtain the required $P \subset P_0$.  In particular, this involves finding the representation
of each geometrical 2-hypercut corresponding to $p\in P_0$ as a real $2$-coboundary. I.e.,
we need to suitably assign each directed 1-simplex over $S$, $e=(s_1,s_2)$,  a real value $x_e$.  
The easiest way to do it is by setting  $x_e$ to be the angle between $s_1$ and $s_2$ with respect to 
$p$, in clockwise direction, normalized by ${1  \over {2\pi}}$.  

All this can obviously be done in polynomial time.
\end{proof}

Theorem~\ref{thm:discrepancy} generalizes to higher dimension without difficulty 
for $d$-simplices, and more generally, for triangulable polytopes over $S$. 
%%%%%%%%%%%%%%%%%%%%%%%%%%%%%%%%%%%%%%%%%%%%%%%%%%%%%
%                                THE END
%%%%%%%%%%%%%%%%%%%%%%%%%%%%%%%%%%%%%%%%%%%%%%%%%%%%%%%
$\mbox{}$ \\ \\
{\bf ACKNOWLEDGMENTS}\\
We would like to thank for valuable comments and discussions to Vladimir Hinich, Gil Kalai, 
Nati Linial, Avner Magen, Jirka Matousek, Roy Meshulam, Tasos Sidiropoulos and Uli Wagner.

\newpage
\bibliographystyle{plain}
%\bibliography{my}

%\newpage
%\section{Appendix}
%%%%%%%%%%%%%%%%%%%%%%%%%%

%

\end{document}